\documentclass{article}
 
\usepackage{etex}
\usepackage{tikz-cd}

\tikzset{
  phantom/.style={draw=none,
    commutative diagrams/every label/.append style={/tikz/auto=false}
  },
}

\def\tikzcdset{\pgfqkeys{/tikz/commutative diagrams}}

\makeatletter

\def\tikzcd@setanchor#1[#2]#3\relax{%
  \ifx\relax#2\relax\else%
    \tikzcdset{@#1transform/.append style={#2},@shiftabletopath}%
  \fi%
  \ifx\relax#3\relax%
    \pgfutil@namelet{tikzcd@#1anchor}{pgfutil@empty}%
  \else%
    \pgfutil@namedef{tikzcd@#1anchor}{.#3}%
  \fi}

\tikzcdset{
  @shiftabletopath/.style={
    /tikz/execute at begin to={%
      \begingroup%
        \def\tikz@tonodes{coordinate[pos=0,commutative diagrams/@starttransform/.try](tikzcd@nodea) %
                          coordinate[pos=1,commutative diagrams/@endtransform/.try](tikzcd@nodeb)}%
        \path (\tikztostart) \tikz@to@path;%
      \endgroup%
      \def\tikztostart{tikzcd@nodea}%
      \tikzset{insert path={(tikzcd@nodea)}}},
    /tikz/commutative diagrams/@shiftabletopath/.code={}},
  start anchor/.code={%
    \pgfutil@ifnextchar[{\tikzcd@setanchor{start}}{\tikzcd@setanchor{start}[]}#1\relax},
  end anchor/.code={%
    \pgfutil@ifnextchar[{\tikzcd@setanchor{end}}{\tikzcd@setanchor{end}[]}#1\relax},
  start anchor=,
  end anchor=,
  shift left/.style={
    /tikz/commutative diagrams/@shiftabletopath,
    /tikz/execute at begin to={%
      \pgfpointnormalised{%
        \pgfpointdiff{\pgfpointanchor{tikzcd@nodeb}{center}}{\pgfpointanchor{tikzcd@nodea}{center}}}%
      \pgfgetlastxy{\tikzcd@x}{\tikzcd@y}%
      \pgfmathparse{#1}%
      \coordinate (tikzcd@nodea) at ([shift={(\pgfmathresult*\tikzcd@y,-\pgfmathresult*\tikzcd@x)}]tikzcd@nodea);%
      \coordinate (tikzcd@nodeb) at ([shift={(\pgfmathresult*\tikzcd@y,-\pgfmathresult*\tikzcd@x)}]tikzcd@nodeb);%
      \tikzset{insert path={(tikzcd@nodea)}}}},
  shift right/.style={
    /tikz/commutative diagrams/shift left={-(#1)}},
  transform nodes/.style={
    /tikz/commutative diagrams/@shiftabletopath,
    /tikz/commutative diagrams/@starttransform/.append style={#1},
    /tikz/commutative diagrams/@endtransform/.append style={#1}},
  shift/.style={
    /tikz/shift={#1},
    /tikz/commutative diagrams/transform nodes={/tikz/shift={#1}}},
  xshift/.style={
    /tikz/xshift={#1},
    /tikz/commutative diagrams/transform nodes={/tikz/xshift={#1}}},
  yshift/.style={
    /tikz/yshift={#1},
    /tikz/commutative diagrams/transform nodes={/tikz/yshift={#1}}}}

\makeatother
\tikzcdset{
  shift left/.default=+0.7ex,
  shift right/.default=+0.7ex}

\newenvironment{tikzar}[1][]{{}\kern-4pt\begin{tikzcd}[ampersand replacement=\&,#1]}%
{\end{tikzcd}\kern-4pt{}}

\usepackage{amssymb}
\usepackage{amsmath}
\usepackage{amsthm}
\usepackage{hyperref}
\usepackage{txfonts}
\usepackage{bbm}
\usepackage{enumerate}
\usepackage{multicol}
\usepackage{rotate}

\usepackage[T1]{fontenc}
\usepackage[utf8]{inputenc}
\usepackage{authblk}

\theoremstyle{plain}
\newtheorem{theorem}{Theorem}[section]
\newtheorem{proposition}[theorem]{Proposition}
\newtheorem{lemma}[theorem]{Lemma}

\theoremstyle{definition}
\newtheorem{definition}{Definition}[section]

\theoremstyle{remark}

\newdimen\proofrulebreadth \proofrulebreadth=.05em
\newdimen\proofdotseparation \proofdotseparation=1.25ex
\newdimen\proofrulebaseline \proofrulebaseline=2ex
\newcount\proofdotnumber \proofdotnumber=3
\let\then\relax
\def\hfi{\hskip0pt plus.0001fil}
\mathchardef\squigto="3A3B
\newif\ifinsideprooftree\insideprooftreefalse
\newif\ifonleftofproofrule\onleftofproofrulefalse
\newif\ifproofdots\proofdotsfalse
\newif\ifdoubleproof\doubleprooffalse
\let\wereinproofbit\relax
\newdimen\shortenproofleft
\newdimen\shortenproofright
\newdimen\proofbelowshift
\newbox\proofabove
\newbox\proofbelow
\newbox\proofrulename
\def\shiftproofbelow{\let\next\relax\afterassignment\setshiftproofbelow\dimen0 }
\def\shiftproofbelowneg{\def\next{\multiply\dimen0 by-1 }%
\afterassignment\setshiftproofbelow\dimen0 }
\def\setshiftproofbelow{\next\proofbelowshift=\dimen0 }
\def\setproofrulebreadth{\proofrulebreadth}

\def\prooftree{
\ifnum  \lastpenalty=1
\then   \unpenalty
\else   \onleftofproofrulefalse
\fi
\ifonleftofproofrule
\else   \ifinsideprooftree
        \then   \hskip.5em plus1fil
        \fi
\fi
\bgroup
\setbox\proofbelow=\hbox{}\setbox\proofrulename=\hbox{}%
\let\justifies\proofover\let\leadsto\proofoverdots\let\Justifies\proofoverdbl
\let\using\proofusing\let\[\prooftree
\ifinsideprooftree\let\]\endprooftree\fi
\proofdotsfalse\doubleprooffalse
\let\thickness\setproofrulebreadth
\let\shiftright\shiftproofbelow \let\shift\shiftproofbelow
\let\shiftleft\shiftproofbelowneg
\let\ifwasinsideprooftree\ifinsideprooftree
\insideprooftreetrue
\setbox\proofabove=\hbox\bgroup$\displaystyle 
\let\wereinproofbit\prooftree
\shortenproofleft=0pt \shortenproofright=0pt \proofbelowshift=0pt
\onleftofproofruletrue\penalty1
}

\def\eproofbit{
\ifx    \wereinproofbit\prooftree
\then   \ifcase \lastpenalty
        \then   \shortenproofright=0pt  
        \or     \unpenalty\hfil         
        \or     \unpenalty\unskip       
        \else   \shortenproofright=0pt  
        \fi
\fi
\global\dimen0=\shortenproofleft
\global\dimen1=\shortenproofright
\global\dimen2=\proofrulebreadth
\global\dimen3=\proofbelowshift
\global\dimen4=\proofdotseparation
\global\count255=\proofdotnumber
$\egroup  
\shortenproofleft=\dimen0
\shortenproofright=\dimen1
\proofrulebreadth=\dimen2
\proofbelowshift=\dimen3
\proofdotseparation=\dimen4
\proofdotnumber=\count255
}

\def\proofover{
\eproofbit 
\setbox\proofbelow=\hbox\bgroup 
\let\wereinproofbit\proofover
$\displaystyle
}%
\def\proofoverdbl{
\eproofbit 
\doubleprooftrue
\setbox\proofbelow=\hbox\bgroup 
\let\wereinproofbit\proofoverdbl
$\displaystyle
}%
\def\proofoverdots{
\eproofbit 
\proofdotstrue
\setbox\proofbelow=\hbox\bgroup 
\let\wereinproofbit\proofoverdots
$\displaystyle
}%
\def\proofusing{
\eproofbit 
\setbox\proofrulename=\hbox\bgroup 
\let\wereinproofbit\proofusing
\kern0.3em$
}

\def\endprooftree{
\eproofbit 
  \dimen5 =0pt
\dimen0=\wd\proofabove \advance\dimen0-\shortenproofleft
\advance\dimen0-\shortenproofright
\dimen1=.5\dimen0 \advance\dimen1-.5\wd\proofbelow
\dimen4=\dimen1
\advance\dimen1\proofbelowshift \advance\dimen4-\proofbelowshift
\ifdim  \dimen1<0pt
\then   \advance\shortenproofleft\dimen1
        \advance\dimen0-\dimen1
        \dimen1=0pt
        \ifdim  \shortenproofleft<0pt
        \then   \setbox\proofabove=\hbox{%
                        \kern-\shortenproofleft\unhbox\proofabove}%
                \shortenproofleft=0pt
        \fi
\fi
\ifdim  \dimen4<0pt
\then   \advance\shortenproofright\dimen4
        \advance\dimen0-\dimen4
        \dimen4=0pt
\fi
\ifdim  \shortenproofright<\wd\proofrulename
\then   \shortenproofright=\wd\proofrulename
\fi
\dimen2=\shortenproofleft \advance\dimen2 by\dimen1
\dimen3=\shortenproofright\advance\dimen3 by\dimen4
\ifproofdots
\then
        \dimen6=\shortenproofleft \advance\dimen6 .5\dimen0
        \setbox1=\vbox to\proofdotseparation{\vss\hbox{$\cdot$}\vss}%
        \setbox0=\hbox{%
                \advance\dimen6-.5\wd1
                \kern\dimen6
                $\vcenter to\proofdotnumber\proofdotseparation
                        {\leaders\box1\vfill}$%
                \unhbox\proofrulename}%
\else   \dimen6=\fontdimen22\the\textfont2 
        \dimen7=\dimen6
        \advance\dimen6by.5\proofrulebreadth
        \advance\dimen7by-.5\proofrulebreadth
        \setbox0=\hbox{%
                \kern\shortenproofleft
                \ifdoubleproof
                \then   \hbox to\dimen0{%
                        $\mathsurround0pt\mathord=\mkern-6mu%
                        \cleaders\hbox{$\mkern-2mu=\mkern-2mu$}\hfill
                        \mkern-6mu\mathord=$}%
                \else   \vrule height\dimen6 depth-\dimen7 width\dimen0
                \fi
                \unhbox\proofrulename}%
        \ht0=\dimen6 \dp0=-\dimen7
\fi
\let\doll\relax
\ifwasinsideprooftree
\then   \let\VBOX\vbox
\else   \ifmmode\else$\let\doll=$\fi
        \let\VBOX\vcenter
\fi
\VBOX   {\baselineskip\proofrulebaseline \lineskip.2ex
        \expandafter\lineskiplimit\ifproofdots0ex\else-0.6ex\fi
        \hbox   spread\dimen5   {\hfi\unhbox\proofabove\hfi}%
        \hbox{\box0}%
        \hbox   {\kern\dimen2 \box\proofbelow}}\doll%
\global\dimen2=\dimen2
\global\dimen3=\dimen3
\egroup 
\ifonleftofproofrule
\then   \shortenproofleft=\dimen2
\fi
\shortenproofright=\dimen3
\onleftofproofrulefalse
\ifinsideprooftree
\then   \hskip.5em plus 1fil \penalty2
\fi
}

\renewcommand{\to}{\xrightarrow{}}
\newcommand{\tto}[1]{\xrightarrow{#1}}

\newcommand{\mono}{\rightarrowtail}
\newcommand{\epi}{\twoheadrightarrow}
\newcommand{\inclusion}{\hookrightarrow}
\newcommand{\mmono}[1]{\stackrel{#1}\rightarrowtail}
\newcommand{\eepi}[1]{\stackrel{#1}\twoheadrightarrow}

\newcommand{\id}{\mathrm{id}}

\newcommand{\EEE}{{\cal E}}

\newcommand{\KKK}{{\cal K}}

\renewcommand{\Bbb}{\mathbb}
\newcommand{\AAa}{{\Bbb A}}
\newcommand{\BBb}{{\Bbb B}}
\newcommand{\CCc}{{\Bbb C}}

\newcommand{\NNn}{{\Bbb N}}

\mathcode`\<="4268 
\mathcode`\>="5269 
\mathchardef\gt="313E 
\mathchardef\lt="313C 

 \def\pushright#1{{
    \parfillskip=0pt            
    \widowpenalty=10000         
    \displaywidowpenalty=10000  
    \finalhyphendemerits=0      
    \leavevmode                 
    \unskip                     
    \nobreak                    
    \hfil                       
    \penalty50                  
    \hskip.2em                  
    \null                       
    \hfill                      
    {#1}                        
    \par}}                      

 \def\qed{\pushright{$\square$}\penalty-700 \smallskip}

{\qed\end{trivlist}}

\newcommand{\beq}{\begin{equation}}
\newcommand{\eeq}{\end{equation}}
\newcommand{\ba}[1]{\begin{array}{#1}}
\newcommand{\ea}{\end{array}}
\newcommand{\bea}{\begin{eqnarray}}
\newcommand{\eea}{\end{eqnarray}}
\newcommand{\bear}{\begin{eqnarray*}}
\newcommand{\eear}{\end{eqnarray*}}

\newcommand{\Id}{{\rm Id}}

\newcommand{\lft}[1]{\overleftarrow{#1}}
\newcommand{\rgt}[1]{\overrightarrow{#1}}

\newcommand{\ladj}[1]{{#1}^\ast}
\newcommand{\radj}[1]{{#1}_\ast}
\newcommand{\adj}[1]{\ladj {#1} \dashv \radj {#1}}
\newcommand{\lkadj}[1]{{#1}^\flat}
\newcommand{\rkadj}[1]{{#1}_\flat}
\newcommand{\kadj}[1]{\lkadj{#1} \dashv \rkadj {#1}}
\newcommand{\lnadj}[1]{{#1}^\sharp}
\newcommand{\rnadj}[1]{{#1}_\sharp}
\newcommand{\nadj}[1]{\lnadj{{#1}} \dashv \rnadj {{#1}}}

\newcommand{\TKL}{\overleftarrow{\KKK}}
\newcommand{\TEM}{\overleftarrow{\EEE}}

\newcommand{\Klm}[2]{{{#1}}_{\overleftarrow{#2}}}

\newcommand{\Emm}[2]{{{#1}}^{\overleftarrow{#2}}}
\newcommand{\Klc}[2]{{{#1}}_{\overrightarrow{#2}}}

\newcommand{\Emc}[2]{{{#1}}^{\overrightarrow{#2}}}
\newcommand{\Kar}[1]{{#1}^{\circlearrowleft}}
\newcommand{\Karm}[2]{{#1}_{\overset{\looparrowleft}{#2}}}
\newcommand{\Karmex}[2]{{#1}_{\looparrowleft}^{{#2}}}
\newcommand{\Karc}[2]{{#1}^{\looparrowright}_{{#2}}}
\newcommand{\Karco}[2]{{#1}_{\overrightarrow{#2}}^{\circlearrowleft}}
\newcommand{\Karmo}[2]{{#1}_{\overleftarrow{#2}}^{\circlearrowleft}}

\newcommand{\Karme}[2]{{#1}^{\overleftarrow{#2}}_{\circlearrowleft}}

\newcommand{\co}[1]{\overline{#1}}

\begin{document}

\title{Quotients in monadic programming:\\
Projective algebras are equivalent to coalgebras}

\author{Dusko Pavlovic\thanks{Supported by AFOSR and NSF.} \hspace{5em} Peter-Michael Seidel\\
{University of Hawaii, Honolulu HI, USA}}

\date{}

\maketitle

\begin{abstract}
In monadic programming, datatypes are presented as free algebras, generated by data values, and by the algebraic operations and equations capturing some computational effects. These algebras are free in the sense that they satisfy just the equations imposed by their algebraic theory, and remain free of any additional equations.  The consequence is that they do not admit quotient types. This is, of course, often inconvenient. Whenever a computation involves data with multiple representatives, and they need to be identified according to some equations that are not satisfied by all data, the monadic programmer has to leave the universe of free algebras, and resort to explicit destructors. We characterize the situation when these destructors are preserved under all operations, and the resulting quotients of free algebras are also their subalgebras. Such quotients are called \emph{projective}. Although popular in universal algebra, projective algebras did not attract much attention in the monadic setting, where they turn out to have a surprising avatar: 
for any given monad, a suitable category of projective algebras is equivalent with the category of coalgebras for the comonad induced by any monad resolution. For a monadic programmer, this equivalence provides a convenient way to implement polymorphic quotients as coalgebras. The dual correspondence of injective coalgebras and all algebras leads to a different family of quotient types, which seems to have a different family of applications. Both equivalences also entail  several general corollaries concerning monadicity and comonadicity.
\end{abstract}


\section{Introduction}\label{Sec:Intro}
\subsection{The story}
\paragraph{Background: Monadic programming.}
Monads are one of functional programmers' favorite tools, and possibly one of the most popular categorical concepts \cite{BentonN:monads,BentonN:impact}. As a type constructor, a monad gives rise to datatypes that capture not only the data values, but also some computational effects of interest \cite{MoggiE:monads,WadlerP:monads}. While this is achieved using a very simple and convenient set of tools, the history of the underlying ideas is convoluted, and the conceptual and technical background of monadic programming covers enough territory of algebra and of category theory to conceal many mathematical mysteries.

The conceptual origin of monadic programming was probably the idea that data structures can be captured as algebraic theories, which goes back to the early days of semantics of computation \cite{ADJ-intro,ADJ-JACM}. The technical origin of monadic programming was then the idea that algebraic theories can be captured as monads, which goes back even further, to the early days of category theory \cite[Ch.~I]{LintonF:LaJolla,ManesE:algt}. The upshot of the view of data-structures-as-algebraic-theories  is that computational datatypes, as  domains of  inductive and recursive definitions, can be viewed as \emph{initial}, or \emph{free}\/ algebras, implementing induction as a universal property. The upshot of the view of algebraic-theories-as-monads in this context is the fact that monads encapsulate and hide behind their standard structure\footnote{Godement introduced monads under the name \emph{standard construction}, standardizing the sheaf cohomology construction \cite{Godement}.} the diverse and often bewildering sets of algebraic operations, and make them available only through uniformly structured monadic combinators. 

The main feature of computational monads is thus their succinct and elegant rendering of inductive datatypes as \emph{free}\/ algebras. But this main feature is perhaps also their main limitation: the quotients of free algebras are not free algebras. 

\paragraph{Problem: Quotient types.}
Whenever a data value can be given by different representatives, its datatype is a quotient. E.g., each rational number can be represented by infinitely many fractions ($\frac 3 7 = \frac{39}{91} = \frac{273}{637} = \cdots$), so the datatype of rationals is a quotient of the datatype of ordered pairs of integers. Sets are a quotient of bags, bags are a quotient of lists, and so on. Identifying the equivalent representatives can be a hard and important computational task, tackled in type theory from the outset, going back to Martin-L\"of, and still further back to Leibniz. Different applications often justify different implementations \cite{Altenkirch-Ghani,AltenkirchT:POPL16}, which vary from simply carrying explicit equivalence relations with \emph{setoids} \cite{BartheG:setoids,PowerJ:setoid}, through carrying \emph{coherent}\/ equivalences with \emph{groupoids}\/  \cite{Hofmann-Streicher,HofmannM:quotient,AltenkirchT:LICS99}, all the way to the rich structure of \emph{homotopy types}\/ \cite{AwodeyS:LICS-HOTT,hottbook}, where the problem of quotients in type theory and the problem of invariants in geometry seem to be solving each other.

The basic idea of monadic programming, to present datatypes as \emph{free}\/ algebras, precludes direct implementations of quotient types, since a quotient of a free algebra is in general not free. This is often viewed as a feature, since polymorphism requires that all data satisfy exactly the same equations, which for algebras means that they should satisfy just the equations imposed by their algebraic theory, and no additional equations. When necessary, the additional equations can be imposed by explicit destructors, but the polymorphic constructions generally do not carry over to such quotients, unless the destructors preserve them. Under which conditions does that happen?

\paragraph{Solution: Projective algebras.}
In the present paper, we study a special family of quotients of free algebras: those that also happen to be their subalgebras. This means that they can be implemented not only by imposing additional equations, but also by adjoining suitable operations, called \emph{projectors}, as described below. An algebra which is both a quotient and a subalgebra of a free algebra, i.e. its retract, is said to be \emph{projective} \cite[\S 82]{Gratzer:book}\footnote{The name is borrowed from theory of modules, where the retracts of free modules are also called projective.}. Since projectors induce precisely the quotients that are preserved by all functors \cite{PareR:absolute}, this aproach seems necessary and sufficient for extending polymorphic constructions from free algebras to their quotients. The equivalence of projective algebras for a monad and all coalgebras for any of the corresponding comonads, claimed in the title of the paper, suggests a link between the problem of polymorphic quotients in monadic programming and the ideas of \emph{comonadic}\/ programming, put forward by several authors \cite{BrookesGeva,UustaluT:comonadic,UustaluT:container-comonad}.  The dual link of injective algebras for a comonad and all algebras for any of the corresponding monads suggests a link between polymorphic quotients of cofree coalgebras and unrestricted quotients of algebras. We proceed to work out these links.

\subsubsection*{Prerequisites} 
This is a paper about categorical semantics of computation, so the prerequisites are mostly categorical. The main background definitions are reproduced in the Appendix. A very succinct overview of the underlying concepts can be found in \cite[Sec.~I.3]{JohnstoneP:stone}.

\subsection{Motivating example}\label{sec:example}
\label{Sec:interference}
Let $\CCc$ be a cartesian closed category, i.e. given with an adjunction
\beq\label{eq:CCC}
\begin{tikzar}{}
\CCc \arrow[bend left = 13]{rr}{A\times }[swap]{\bot} \&\& \CCc  \arrow[bend left = 13]{ll}{A\Rightarrow} 
\end{tikzar}
\eeq
for every object $A$. Fix an object $S$ as a \emph{state space}, and consider the monad
\bea
\lft S :\  \CCc & \to & \CCc\label{eq:state}\\
X & \mapsto & S\Rightarrow (S\times X) \notag
\eea
induced by the adjunction \eqref{eq:CCC} instantiated to $S$. As explained in \cite{MoggiE:monads}, the category of free $\lft S$-algebras $\Klm \CCc S$ captures computations with explicit state, or with side effects. A computation over  the inputs of type $A$, the outputs of type $B$, and the states of type $S$ is presented as a free algebra homomorphism $f \in \Klm \CCc S(A, B)$, which can be conveniently viewed as a $\CCc$-morphism in the form $A \tto{\ \ f\ \ } S\Rightarrow (S\times B)$  \cite{KleisliH}. This computation thus maps every input $a\in A$ to a function $S\tto{\ f(a)\ } S\times B$, determining at each state $s\in S$ a next state, and an output. Equivalently, such morphisms can be viewed in the transposed form $S\times A \tto{\ \ f'\ \ } S\times B$, assigning to each state and every input a next state and an output. This transposed form of homomorphisms between free algebras will turn out to be more convenient for the purposes of this paper. In the case of the state monad $\lft S$, such homomorphisms capture Mealy machines  \cite[Sec.~2.7(a)]{Mealy-bonsague,Mealy-helle,Hopcroft-Ullman:book}. 

Towards a more concrete example, consider the following model of data release policies from \cite{PavlovicD:privacy}. Suppose that a Mealy machine $S\times A \tto f S\times B$ models a database. $S$ are its states, $A$ are the inputs  (insertions, queries, \ldots) that the users may enter, and $B$ are the outputs supplied by the database. A stateless map $A\tto g B$ can be thought of a rudimentary deterministic channel, just mapping data of type $A$ to data of type $B$. Since there are multiple users,  there may be privacy policies, and authorization policies that need to be implemented. A privacy policy can be viewed as a map $S\times B\tto \psi S\times B$, which projects any output $b\in B$ of the database at a state $s\in S$ to a sanitized, public component $\psi(s,b)$ of the state and the output of the database, and filters out all private data. An authorization policy can be similarly viewed as a map $S\times A\tto \varphi S\times A$, which projects any database state $s\in S$, and any user input $a\in A$ (including the relevant credentials) to the authorized component $\varphi(s,a)$ of the state and the input. Since the public components should not contain any private residue to be filtered out in a second run of $\psi$, and the authorized components should not contain any unauthorized residue, the policies should be idempotent, i.e. satisfy the equations
\[\varphi\circ \varphi = \varphi \qquad \qquad \psi\circ \psi = \psi\]
Such equations define projectors.  There are at least two different ways to interpret projectors as policies. One is to view them as \emph{policy specifications}. The database $S\times A\tto f S\times B$ then implements the policies $S\times B\tto \psi S\times B$ and $S\times A\tto \varphi S\times A$ if it satisfies the equation
\bea\label{eq:compliant}
f & = & \psi \circ f \circ \varphi
\eea
which is easily seen to be equivalent to the pair of equations
\bea\label{eq:compliant-too}
  \psi\circ f\  =\  f\  =\  f\circ \varphi  
  \eea
In other words, a \emph{compliant}\/ database only ever supplies public data, and only ever permits authorized requests.

A different view to interpret projectors as policies is to view them as \emph{policy implementations}. The database $S\times A\tto f S\times B$ then does not implement the policies itself, but needs to be precomposed with the authorization policy $S\times A\tto \varphi S\times A$ and  postcomposed with the privacy policy $S\times B\tto \psi S\times B$. However, since each of these policies regulates the same data release, the database $S\times A\tto f S\times B$ is \emph{consistent}\/ with the policies $S\times B\tto \psi S\times B$ and $S\times A\tto \varphi S\times A$ if it satisfies the equation
\bea\label{eq:consistent}
\psi \circ f & = & f \circ \varphi
\eea
It is obvious that compliance implies consistency. A consistent database, however, does not have to be compliant. The reason is that a consistent database does not have to implement the policies itself, but it requires separate policy implementations at the input and at the output. On its own, such a database may accept unauthorized requests and it may supply private data. Its consistency means that \emph{if}\/ we make sure that no unauthorized requests are submitted, \emph{then}\/ we can be sure that no private responses will be supplied, and \emph{vice versa}. More precisely, a database $f$ is consistent with the policies $\varphi$ and $\psi$ whenever a request consistent with $\varphi$ results in a response consistent with $\psi$, and all responses consistent with $\psi$ can be obtained on requests consistent with $\varphi$. Since the two policies thus precisely enforce each other on a database consistent with them, they implement the same data release process on this database, which can thus be implemented either as an authorization policy, or as a privacy policy. 

Lifting this distinction between compliance and consistency from databases and state monads, to a distinction between two types of homomorphisms between projective algebras, we arrive at the main results of the paper.

\subsection{The setting and the result}
Every adjunction $\adj F: \BBb\to \AAa$ determines 
\begin{itemize}
\item the monad $\lft F = \radj F \ladj F : \AAa \to \AAa$, with the induced categories of
\begin{itemize}
\item free algebras $\Klm \AAa F$, 
\item projective algebras $\Karmo \AAa F$, and
\item all algebras $\Emm \AAa F$,
\end{itemize}
and on the other hand
\item the comonad $\rgt F = \ladj F \radj F: \BBb \to \BBb$, with the induced categories of
\begin{itemize}
\item cofree coalgebras $\Klc \BBb F$, 
\item injective coalgebras $\Karco \BBb F$, and
\item all coalgebras $\Emc \BBb F$.
\end{itemize}
\end{itemize}
The other way around, given a monad $\lft F :\AAa \to \AAa$, any adjunction $\adj F : \BBb \to \AAa$ such that $\lft F = \radj F \ladj F$ is a \emph{resolution}\/ of $\lft F$; and given a comonad $\rgt F :\BBb \to \BBb$, any adjunction such that $\rgt F = \ladj F \radj F$ is a resolution of $\rgt F$. Each of the above categories defined for a monad (resp. for a comonad) gives its resolution. The definitions are standard, and can be found in the Appendix. In this paper, we introduce the categories of
\begin{itemize}
\item projective algebras with \emph{consistent}\/ morphisms $\Karm \AAa F$, and
\item injective coalgebras with \emph{consistent}\/ morphisms $\Karc \BBb F$.
\end{itemize}
We prove the following equivalences of categories
\beq\label{eq:eqs}
\Emm \AAa F \simeq \Karc \BBb F \qquad \qquad \Emc \BBb F \simeq \Karm \AAa F 
\eeq
under the assumption that the categories $\AAa$ and $\BBb$ are Cauchy\footnote{The habit of attributing categorical concepts to XIX century mathematicians and philosophers has styled the terminology in some parts of category theory.} complete, which means that they split idempotents. This is a very mild assumption, since the Cauchy completion, the weakest of all categorical completions, is easy to construct for any category, as spelled out in Prop.~\ref{prop:splitting}.

\subsection*{Overview of the paper}
Sec.~\ref{Sec:AlgCoalg} begins with a discussion about projectors, and analyzes projectors over free algebras, which determine projective algebras. In Sec.~\ref{Sec:ProjCoalg} we state and prove the main theorem, showing that projectors over free algebras correspond to coalgebras. We also state the dual version, which says that projectors over cofree coalgebras correspond to algebras. In Sec.~\ref{Sec:Application}, we return to the motivating example from Sec.~\ref{sec:example}, and analyze the different categorical formalisations of data release policies. Sec.~\ref{Sec:Closing} closes the paper with comments about the related past work, and about the future work.

\section{Projectors over algebras and coalgebras}\label{Sec:AlgCoalg}

\subsection{Projectors in general}
Consider an equalizer and coequalizer diagram
\[\begin{tikzar}{}
 E \ar[tail]{r}{i} \& A \ar[shift left = .75ex]{r}{\varphi} \ar[shift right = .75ex]{r}[swap]{\id} \& A \ar[two heads]{r}{q} \&
 Q
\end{tikzar}
\]
for an arbitrary endomorphism $\varphi$. Intuitively, the equalizer $E$ consists of the \emph{fixed points}\/ of $\varphi$, whereas the coequalizer $Q$ is the quotient where each element of $A$ is identified with all of its direct and inverse images along $\varphi$, which together form its \emph{orbit}. The obvious map $E\to Q$ maps each fixed point into a unique orbit; but some orbits may not contain any fixed points. We are interested in the situation when each orbit does contain a fixed point, so that each equivalence class from $Q$ has a canonical representative in $E$. This means that the iterated applications of $\varphi$ push each element of $A$ along its orbit towards a fixed point. It can be shown that this situation is characterized by the requirement that the following countably extended diagram commutes
\[
\begin{tikzar}{}
A \ar[shift left = .75ex]{r}{\varphi} \ar[shift right = .75ex]{r}[swap]{\id} \& A \ar{r}{\varphi} \& A \ar{r}{\varphi} \& \cdots
\end{tikzar}
\]
In terms of elements, this means that for every $x\in A$ there is some $n\in \NNn$ such that $\varphi^{n+1}(x) = \varphi^n(x)$. In other words, $\varphi$ thus equips $A$ with the structure of a \emph{forest}, where the equivalence classes that form $Q$ are the component trees, and the elements of $E$ are their roots. Projectors are the special case of this situation, where already the diagram 
\[\begin{tikzar}{}
A \ar[shift left = .75ex]{r}{\varphi} \ar[shift right = .75ex]{r}[swap]{\id} \& A\ar{r}{\varphi} \& A \end{tikzar}\]
commutes. The forest thus reduces to a shrub, where each component may branch at the root as wide as it likes, but can only grow one layer tall, and cannot grow any further branches. The following summarizes \cite[Sec.~IV.7.5]{GrothendieckA:SGA4}.

\begin{proposition}\label{prop:projectors}
For any endomorphism $\varphi$ with an epi-mono factorization
\[
\begin{tikzar}{}
A \arrow{rr}{\varphi} \arrow[twoheadrightarrow]{dr}[swap]{q} \&\& A\\ \& B \arrow[rightarrowtail]{ur}[swap]{i} 
\end{tikzar}
\]
the following statements are equivalent:
\begin{enumerate}[(a)]
\item $\varphi\circ \varphi = \varphi$
\item $q\circ i = \id$
\item $i$ is an equalizer and $q$ is a coequalizer of $\varphi$ and the identity
\[\begin{tikzar}{}
\& B \ar[tail]{dl}[swap]{i} 
\\ A \ar[shift left = .75ex]{rr}{\varphi} \ar[shift right = .75ex]{rr}[swap]{\id} \&\& A \ar[two heads]{ul}[swap]{q} 
\end{tikzar}
\]
\end{enumerate}
\end{proposition}

\begin{definition}\label{def:splitting}
An endomorphism $A\tto\varphi A$ is called a \emph{projector}\/ (or \emph{idempotent}) if $\varphi\circ \varphi = \varphi$. Its \emph{splitting}\/ is an object $B$ with a pair of arrows $\begin{tikzar}{}
A \arrow[bend left = 15,twoheadrightarrow]{r}{q}  \&  B \arrow[bend left = 15,rightarrowtail]{l}{i} 
\end{tikzar}$ such that $i\circ q = \varphi$ and  $q\circ i = \id$. More succinctly, we write $A\stackrel{\varphi}\looparrowright B$.
\end{definition}
Since the projectors and their splittings are defined by equations, every functor must preserve them. Since a splitting consists of an equalizer and a coequalizer, it is an equalizer and a coequalizer that must be preserved by all functors.
\begin{definition}
A categorical property is called \emph{absolute}\/ when it is preserved by all functors. A category that has all absolute limits and absolute colimits is said to be \emph{Cauchy complete}.
\end{definition}
It follows from the results of \cite{PareR:absolute}, as well as from the different approach in \cite[Sec.~I.6.5]{BorceuxF:handbook}, that all absolute limits and colimits boil down to splittings.
\begin{proposition}\label{prop:splitting}
For any category $\CCc$ the following statements are equivalent\begin{enumerate}[(a)]
\item $\CCc$ is Cauchy complete,
\item all projectors split in $\CCc$,
\item the obvious embedding $\CCc \inclusion \Kar \CCc$ is an equivalence, where
\end{enumerate}
\bear
\rvert\Kar \CCc\rvert & = & \coprod_{A\in |\CCc|} \{A\tto\varphi A\ |\ \varphi\circ \varphi = \varphi\} \\
\Kar\CCc (A\tto\varphi A, B\tto\psi B) & = & \left\{f\in \CCc(A,B)\ \Big|\ \ \ \begin{tikzar}[row sep=1.8em,column sep=1.8em]
 A \ar{r}{f} \ar{d}{\varphi} \& B  \\ 
 A \ar{r}{f}\& B \ar{u}{\psi}
\end{tikzar}\ \right\}
\eear
\end{proposition}
The absolute completion $\Kar \CCc$ is sometimes called \emph{Karoubi envelope}\/ of $\CCc$ \cite[Sec.~IV.7.5]{GrothendieckA:SGA4}. Two categories are \emph{Morita equivalent}\/ when their Cauchy completions are equivalent \cite[Thm.~7.9.4]{BorceuxF:handbook}. Note that the condition $\psi \circ f\circ \varphi = f$ is equivalent with the requirement that both $f\circ \varphi = f$ and $\psi \circ f = f$ are valid.

\paragraph{\sc \textbf{Assumption.}} In the rest of this paper, we assume that each of the categories under consideration is \textbf{\emph{Cauchy complete}}, i.e. that projectors split in it. Any category $\CCc$ that does not fulfill this assumption should be replaced by its Karoubi envelope $\Kar\CCc$, described in Prop.~\ref{prop:splitting}\textit{(c)}.

\subsection{Projective algebras over free algebras}
In homological algebra, projective modules are usually defined as direct summands of free modules. In the terminology of the preceding section, this means that they arise by \emph{splitting the projectors}\/ over free modules. 
It is natural to define projective algebras in a similar way: as projectors over free algebras \cite[\S 82]{Gratzer:book}.

In the usual (Kleisli) view of the category of free algebras \cite{KleisliH}, reproduced in the Appendix, a morphism $f\in \Klm \AAa F(A,B)$ is a morphism $A \tto f \lft F B$ in $\AAa$, and its composition with $g\in \Klm \AAa F(B,C)$, which is $B\tto g \lft F C$ in $\AAa$ is defined by
\bear
g\circ_{\lft F} f & = & \left(A\tto{\ f\ } \lft FB \tto{\lft F g} \lft F\lft F C \tto{\ \mu\ } \lft F C\right) \ \
 \in \ \  \Klm \AAa F (A,C)
\eear
A projector $\varphi \in \Klm \AAa F(A, A)$ over the free algebra generated by $A$ is thus an $\AAa$-morphism $A\tto \varphi \lft F A$ such that $\varphi \circ_{\lft F} \varphi = \mu_A \circ \lft F \varphi \circ \varphi = \varphi$. 
As they expand, the calculations with projectors in the Kleisli form of category $\Klm \AAa F$ do get increasingly clumsy.

When a monad $\lft F:\AAa \to \AAa$ is induced by an adjunction $\adj F :\BBb \to \AAa$ so that $\lft F = \radj F \ladj F$, then the category of free algebras can be equivalently defined by
\bea\notag
|\Klm \AAa F| & = & |\AAa|\\
\Klm \AAa F (X, Y) & = & \BBb(\ladj F X, \ladj F Y)\label{def:kleis}
\eea
It is easy to check that the natural bijections 
\bear
\AAa(X, \radj F\ladj F Y) & \cong &  \BBb(\ladj F X, \ladj F Y)
\eear
map the Kleisli composition of the morphisms in $\AAa(X, \lft F Y)$ to the ordinary composition of their adjunction transposes in $\BBb(\ladj F X, \ladj F Y)$. When the homomorphisms between free $\lft F$-algebras are presented as the elements of $\BBb(\ladj F X, \ladj F Y)$,  then a projector $\varphi \in \Klm \AAa F(X, X)$ is just a $\BBb$-morphism $\ladj F X\tto \varphi \ladj F X$ such that $\varphi\circ \varphi = \varphi$. The category of projective $\lft F$-algebras and \emph{compliant}\/ homomorphisms (explained in Sec.~\ref{sec:example}) is thus defined as follows:
\bea
|\Karmo \AAa F| & = & \coprod_{X\in |\AAa |} \Big\{\varphi \in \BBb(\ladj F X, \ladj F X)\ \Big| \ \varphi\circ \varphi = \varphi\Big\}\notag \\
\Karmo \AAa F (\varphi, \psi) & = & \left\{h \in \BBb(\ladj FX, \ladj F Y)\ \Big|\ \ 
\begin{tikzar}[row sep=1.8em,column sep=1.8em]
\ladj F X \ar{r}{h} \ar{d}{\varphi} \& \ladj F Y  \\ 
\ladj F X \ar{r}{h}\& \ladj F Y \ar{u}{\psi}
\end{tikzar}\ 
\right\}\label{eq:compliant-form}
\eea
where $\psi\in  \BBb(\ladj F C, \ladj F C)$ is another projective algebra, viewed as a projector over the free algebra generated by $C\in |\AAa|$. 

We note that the category $\Klm \AAa F$, defined in \eqref{def:kleis}, is clearly isomorphic with the usual Kleisli category, recalled in Appendix A. It follows that the category $\Klm \AAa F$, with its projectors, its compliant homomorphisms, and its consistent homomorphisms in the next section, only depends on the monad $\lft F$, and not on the resolution $\adj F$. The category $\BBb$ and the adjoint $\ladj F$ are used in the above definitions only for convenience. The relevant concepts could be equivalently defined within in the standard Kleisli category, and that definition is in fact the special case of the above, for the Kleisli resolution from Def.A.4. But our results would look seem significantly more complicated in that framework.

\subsection{Projective algebras among all algebras}
The category of free $\lft F$-algebras $\Klm \AAa F$ embeds fully and faithfully into the category $\Emm \AAa F$ of all $\lft F$-algebras by the functor
\beq\label{eq:M}
\begin{tikzar}[row sep=1.5em,column sep=4em]
\Klm \AAa F \arrow[hook]
{r}{M} \& \Emm \AAa F  \end{tikzar}
\eeq
defined by
\[
\prooftree
X \in \big|\Klm \AAa F\big|
\justifies
\left(\lft F \lft F X\tto \mu \lft F X\right) \in \big|\Emm \AAa F\big|
\endprooftree
\quad \mbox{and} \quad  
\prooftree
\left(\ladj F X \tto{h} \ladj F Y\right) \in \Klm \AAa F (X, Y)
\justifies
\left(\lft F X \tto{\radj F h}\lft F Y\right)\ \in\ \Emm \AAa F\big(\mu_X, \mu_Y\big)
\endprooftree
\]
where $\radj F h$ is an algebra homomorphism because $\mu = \radj F \varepsilon$, and the naturality of $\varepsilon$ thus implies $\radj F h \circ \mu_X = \mu_Y \circ \lft F \radj F h$. Since the projectors in $\Emm \AAa F$ split whenever they split in $\AAa$, as assumed here, the embedding $M$ has a unique extension $M^\circlearrowleft$,
\beq\label{eq:Mcircle}
\begin{tikzar}[row sep=1.5em,column sep=4em]
\Karmo \AAa F \arrow[hook,dashed]
{r}{M^\circlearrowleft} \& \Emm \AAa F \\
\Klm \AAa F \arrow[hook]
{ur}[swap]{M} \arrow[hook]{u} \end{tikzar}
\eeq
which maps each $\varphi \in \lvert \Karmo \AAa F\rvert$ to a splitting $\alpha$ of the projector $\radj F \varphi \in \Emm \AAa F(\mu_X, \mu_X)$
\beq\label{eq:Kvarphi}
\begin{tikzar}{}
\lft F \lft F X \ar[two heads]{r}{\lft F q} \arrow{d}[swap]{\mu} 
\& \lft F A \arrow{d}{\alpha}
\ar[tail]{r}{\lft F i}  \& \lft F \lft F X \arrow{d}{\mu}
\\
\lft F X \ar[bend right = 20]{rr}[swap]{\radj F \varphi}
\ar[two heads]{r}{q}  
\& A  \ar[tail]{r}{i}   \&  \lft FX
\end{tikzar}
\eeq
%
%
%
An $\lft F$-algebra $\lft F A \tto \alpha A$ is thus projective if it is a retract of some free algebra $\lft F\lft FX \tto\mu \lft FX$, i.e. if there are an $X\in |\AAa|$ and some homomorphisms $q \in \Emm \AAa F (\mu_X, \alpha)$ and $i \in \Emm\AAa F(\alpha, \mu_X)$ such that $q\circ i = \id$.

It turns out, however, that each projective algebra $\lft F A \tto \alpha A$ is not just a retract of a free algebra over some $X\in |\AAa|$, but a retract of the free algebra over its own carrier $A$. 

\begin{proposition}\label{prop:charProjective}
An algebra $\lft F A \tto \alpha A$ is projective if and only if there is a unique algebra homomorphism $\co\alpha\in  \Emm\AAa F(\alpha, \mu_A)$ such that $\alpha \circ \co\alpha = \id$.
\[
\begin{tikzar}{}
\lft F A \ar[tail]{r}{\lft F\co\alpha} \arrow{d}[swap]{\alpha}  
\&  \lft F \lft F A \ar[two heads]{r}{\lft F \alpha} \arrow{d}{\mu_A} 
\& \lft F A \arrow{d}{\alpha} 
\\
A \ar[tail]{r}[swap]{\co\alpha}   
\&  \lft F A 
\ar[two heads]{r}[swap]{\alpha}  
\& A 
\end{tikzar}
\]
\end{proposition}

\begin{proof} To construct $\co\alpha\in  \Emm\AAa F(\alpha, \mu_A)$, extend the algebra homomorphism $q\in \Emm \AAa F(\mu_X, \alpha)$ to $\lft F(q\circ \eta) \in \Emm \AAa F(\mu_X, \mu_A)$ and precompose with $i \in \Emm \AAa F(\alpha, \mu_X)$. Hence $\co\alpha=\lft F q \circ \lft F \eta \circ i \in \Emm \AAa F (\alpha, \mu_A)$, as displayed in the middle row of the following diagram.
\[
\begin{tikzar}{}
\lft F A \ar{d}[swap]{\alpha}\ar{r}{\lft F i} \&\lft F \lft F X \ar{d}{\mu} \ar{r}{\lft F\lft F \eta} \& \lft F\lft F\lft F X \ar{d}{\mu} \ar{r}{\lft F\lft F q} \& \lft F \lft F A \ar{d}{\mu}\\
A \ar{r}[swap]{i} \& \lft F X \ar{r}{\lft F \eta} \ar{dr}[swap]{\id} \&\lft F \lft F X \ar{r}[swap]{\lft F q} \ar{d}{\mu} \& \lft F A \ar{d}{\alpha} 
\\
\&\&\lft F X \ar{r}[swap]{q} \& A
\end{tikzar}
\]
The commutativity of the upper three squares implies that $\co\alpha$ is an $\lft F$-algebra homomorphism. The commutativity of the lower square and the triangle implies that $\alpha \circ{\co\alpha} = q\circ i = \id$.

To see that $\co\alpha$ is unique, i.e. that it is the only way to display $\lft FA \tto \alpha A$ as a subalgebra of $\lft F\lft FA \tto\mu \lft FA$, note that the composite homomorphism $\co\alpha \circ \alpha$ is a projector on the free algebra $\lft F \lft FA \tto \mu \lft FA$, and that $\lft FA \eepi \alpha A$ is the splitting of this projector.
\[
\begin{tikzar}{}
\lft F \lft F A \ar[two heads]{r}{\lft F \alpha} \arrow[two heads]{d}[swap]{\mu} 
\& \lft F A \ar[tail]{r}{\lft F{\co\alpha}} \arrow[two heads]{d}[swap]{\alpha}  \&\lft F\lft F A \arrow[two heads]{d}{\mu} 
\\
\lft F A \ar[two heads]{r}[swap]{\alpha}  
\& A \ar[tail]{r}[swap]{{\co\alpha}}   
\&  \lft F A 
\end{tikzar}
\]
While the splitting of a projector is unique up to an isomorphism, fixing one component of the splitting determines the other one on-the-nose: since $\alpha$ is an epi, $\co\alpha_0\circ \alpha = \co\alpha_1\circ \alpha$ implies that $\co\alpha_0=\co\alpha_1$.
\end{proof}

The preceding proposition thus says that every projective algebra $\lft FA\eepi \alpha A$ has a unique embedding $A\mmono {\co\alpha} \lft FA$ into the free algebra $\mu_A$ over its carrier $A$. With no loss of generality, the full subcategory $\Karme \AAa F$ of the (Eilenberg-Moore) category $\Emm A F$ of all $\lft F$-algebras spanned by the projective $\lft F$-algebras can thus be viewed in the form
\bea\notag
\lvert\Karme \AAa F\rvert & = & \left\{\left(\lft FA \tto\alpha A\right) \in \lvert \Emm \AAa F\rvert\ \big|\ \exists{\co\alpha} \in \Emm \AAa F (\alpha, \mu_A).\  \alpha \circ{\co\alpha} = \id_A\right\}\\
& \cong & \coprod_{\alpha\in | \Emm \AAa F|} \Big\{{\co\alpha}\in \Emm \AAa F (\alpha, \mu_A)\ \big|\ \alpha \circ{\co\alpha} = \id_A\Big\}\notag\\ 
\Karme \AAa F\left(\alpha, \gamma\right) & = & \left\{ f \in \AAa(A,C)\ \Big|\ \ \begin{tikzar}[row sep=1.8em,column sep=1.8em]
 \lft F A \ar{d}{\alpha} \ar{r}{\lft F f}
 \& \lft FC \ar{d}[swap]{\gamma} \\
A \ar{r}[swap]{f}  \& C 
\end{tikzar}\  
\right\} \label{eq:EM-proj}
\eea

\begin{lemma}\label{lemma:incl-proj}
Let $\lft FA\eepi \alpha A$ be a projective $\lft F$-algebra and $A\mmono{\co\alpha} \lft FA$ the $\lft F$-algebra monomorphism, as constructed in Prop.\ref{prop:charProjective}, that makes $\alpha$ into a subalgebra of the free algebra $\mu_A$. Then the transpose $\ladj F A\tto{\co\alpha'} \ladj F A$ of $A\tto{\co\alpha} \radj F\ladj F A$
is idempotent\footnote{If the category of free algebras $\Klm \AAa F$ is presented in the Kleisli form, then $A\tto{\co\alpha} \lft FA$   itself is an idempotent morphism, and thus a projector in it.}, and thus a projector in $\Klm \AAa F$. Moreover
\bea\label{eq:ten}
\co\alpha\circ \alpha & = & \radj F{\co\alpha}'
\eea
\end{lemma}

\begin{proof} The fact that $\co\alpha'\circ \co\alpha ' =  \co\alpha'$ can be seen by transposing the following diagram along the adjunction $\adj F$.
\[
\begin{tikzar}{}
A \ar{r}{\co\alpha}  \ar{dr}[swap]{\id}
\& \lft F A \ar{r}{\lft F{\co\alpha}} \arrow{d}[swap]{\alpha}  \&\lft F\lft F A \arrow{d}{\mu} 
\\
\& A \ar{r}[swap]{{\co\alpha}}   
\&  \lft F A 
\end{tikzar}
\]
Equation \eqref{eq:ten} also follows from the commutativity of the square on the above diagram, and the observation that
\[\mu \circ \lft F \co\alpha = \radj F \varepsilon \circ \radj F \ladj F \co\alpha = \radj F\left(\varepsilon \circ \ladj F \co\alpha \right) = \radj F \co\alpha '
\]
\end{proof}

\begin{lemma}\label{lemma:projmorphisms}
Let $\lft FA \tto \alpha A$ and $\lft FC \tto \gamma C$ be projective $\lft F$-algebras, with the $\lft F$-algebra monomorphisms  $A\mmono{\co\alpha}\in \lft F A$ and $C\mmono{\co\gamma}\lft F C$ including them into the free $\lft F$-algebras $\mu_A$ and $\mu_C$, respectively, as in Prop.~\ref{prop:charProjective}. 
Then every $\lft F$-algebra homomorphism $f\in \Emm \AAa F(\alpha, \gamma) = \Karme \AAa F(\alpha, \gamma)$ induces the homomorphism $H f = {\co\gamma}'\circ \ladj F f = \ladj F f \circ{\co\alpha}' \in \Karmo\AAa F (\co\alpha', \co\gamma')$, which is \emph{compliant} in the sense of \eqref{eq:compliant}, since
\bear
f\circ \alpha = \gamma \circ \lft F f & \iff & \co\gamma' \circ H f \circ{\co\alpha}' = H f
\eear
\end{lemma}

\begin{proposition}\label{prop:eqproj}
There is an equivalence of categories
\beq\label{eq:K}
\begin{tikzar}{}
\Karme \AAa F \arrow[bend left = 7]
{rrr}{H}[swap]{\simeq} \&\&\& \Karmo \AAa F   \arrow[bend left = 7]
{lll}{K} 
\end{tikzar}
\eeq
where the object part of the functor $H$ is defined using Lemma~\ref{lemma:incl-proj}
\[
\prooftree
\left(\alpha \mmono{\co\alpha} \mu_A \right) \in \big|\Karme \AAa F\big|
\justifies
H\alpha = \left<A, \ladj FA \tto{\co\alpha'} \ladj F A\right> \in \big|\Karmo \AAa F\big|
\endprooftree
\]
whereas the arrow part is defined in Lemma~\ref{lemma:projmorphisms}.
\[\prooftree
{f}  \in \Karme \AAa F (\alpha, \gamma)
\justifies
Hf = \co\gamma'\circ \ladj F f = \ladj F f \circ{\co\alpha}' \in \Karmo\AAa F (H\alpha , H\gamma) 
\endprooftree
\]
The functor $K$, on the other hand, is the factorization of the functor $M^\circlearrowleft: \Karmo \AAa F \to \Emm \AAa F$ from \eqref{eq:Mcircle} through the inclusion $\Karme \AAa F \inclusion \Emm \AAa F$. More precisely, its object and the arrow parts
\[
\prooftree
\left< X, \ladj FX \tto \varphi\ladj FX\right> \in \big|\Karmo \AAa F\big|
\justifies
\left(\lft FA \tto{K\varphi } A\right) \in \big|\Karme \AAa F\big|
\endprooftree
\quad \mbox{and} \quad  
\prooftree
{h}  \in \Karmo \AAa F (\varphi, \psi)
\justifies
Kh\ \in\ \Karme \AAa F\big(K\varphi, K\psi\big)
\endprooftree
\]
are defined by the projector splittings
\beq\label{eq:spliteq}
\begin{tikzar}{}
\hspace{1em}\mu_X \arrow[loop left]{u}[swap]{\radj F \varphi}
\arrow[bend left = 12,twoheadrightarrow]{rr}{q}  
\ar{dd}[swap]{\radj F h} \&\&  K\varphi \arrow[bend left = 12,rightarrowtail]{ll}{i} \ar[dashed]{dd}{Kh}\\ 
\\
\hspace{1em}  \mu_Y 
\arrow[loop left]{u}[swap]{\radj F \psi} 
\arrow[bend left = 12,twoheadrightarrow]{rr}{p}  
\& \& K\psi \arrow[bend left = 12,rightarrowtail]{ll}{j}
\end{tikzar}
\eeq
\end{proposition}

\begin{proof} Both functors are clearly well defined. Towards the isomorphism $HK\varphi \cong \varphi$, we first split the projector $\mu_X\tto{\radj F\varphi}\mu_X$ in $\Karme \AAa F$ to get the subalgebra $K\varphi = \left(\lft FA \tto{\alpha} A\right)$ of the free algebra $\lft F\lft F X\tto{\mu} \lft FX$, and then construct the homomorphism $A \tto{\co\alpha} \lft FA$, like in Prop.~\ref{prop:charProjective}, to get the projector $HK\varphi = H\alpha = \co\alpha'$. The isomorphism $\varphi \cong HK\varphi$ in $\Karmo\AAa F$ is realized by the transpose $\ladj F A \tto{i'}\ladj FX$ of the monic component $A\tto i \radj F\ladj FX$ of the splitting $\radj F \varphi = \left(\ladj F\radj FX\eepi{q} A \mmono{i}\ladj F\radj F X\right)$ in \eqref{eq:Kvarphi}. The fact that $i'$ is a projector homomorphism from $\co \alpha' = HK \varphi$ to $\varphi$ in $\Karmo\AAa F$ boils down to the equations $\varphi \circ i' = i' = i' \circ \co \alpha ' $. To see that the first one holds, take a look at the adjunction transposes of its two sides:
\[
\radj F \varphi \circ i\ =\ i\circ q\circ i = i 
\]
To see that the second equation holds, consider the following diagram:
\[
\begin{tikzar}{}
A \ar{ddr}[swap]{i}  
\ar{r}[swap]{i} \ar[bend left = 20]{rrr}{\co \alpha} \&
\lft F X \ar{d}[swap]{q} 
\ar[bend left = 20]{dd}{\radj F\varphi} 
\ar{r}[swap]{\eta} 
\& 
\lft F\lft F X \ar[bend left = 20]{dd}[swap]{\lft F \radj F\varphi}
\ar{r}[swap]{\lft F q} \& \lft F A \ar{ldd}{\lft F i}\\
\& A \ar{d}[swap]{i}\\
\& \lft F X \ar{r}{\eta} 
\&\lft F \lft F X \ar[bend left = 20]{l}{\mu}
\end{tikzar}
\]
The two paths around this diagram correspond to the transposes of the two sides of the second equation. Hence $i' \in \Karmo\AAa F(\co\alpha',\varphi)$.  To show that $i'$ is an isomorphism in $\Karmo\AAa F$, consider
\bea\label{eq:r}
i'' & = & \big(\ladj F X \tto{\ladj F \eta} \ladj F \radj F \ladj F X \tto{\ladj F q} \ladj F A \tto{\co\alpha'} \ladj FA\big)
\eea
The equations
\beq\label{eq:inverses} i'' \circ i' = \co \alpha' \qquad\mbox{ and } \qquad i' \circ i'' = \varphi
\eeq
follow directly from the definitions. Since this immediately implies $i''\circ \varphi = i'' = \co\alpha' \circ i''$, it follows that  $i''\in\Karmo\AAa F(\varphi, \co\alpha')$. Equations \eqref{eq:inverses} mean that $i'$ and $i''$ are each other's inverses in $\Karmo\AAa F$. Hence $HK\varphi \cong \varphi$, where $HK\varphi = \co\alpha'$.

The isomorphism $KH\alpha \cong \alpha$ may seem surprising. How can the functor $Hf = \co\gamma '\circ \ladj F f = \ladj F f \circ{\co\alpha}'$ be faithful when $\ladj F$ in general does not have to be faithful? The answer is in the following diagram: 
\[
\begin{tikzar}[row sep=2.5em,column sep=1.8em]
\lft FA  \arrow{rr}[swap]{\lft Ff} \arrow[bend left = 20]{rrrr}{\radj FHf = \radj F \co\gamma ' \circ \radj F\ladj F f}  \ar[bend right = 20]{dd}[swap]{\radj F H\alpha = \radj F\co\alpha'}  \ar[two heads,thick]{d}{\alpha} \&\&  \lft FC \arrow[two heads]{d}{\gamma} \arrow{r}[swap]{\gamma} \& C \arrow{r}[swap]{\co\gamma} \& \lft F C\arrow[two heads]{d}[swap]{\gamma} \arrow[bend left = 20]{dd}{\radj F H\gamma =\radj F\co\gamma'}\\
A \arrow[dashed,thick]{rr}{f} \arrow[tail]{d}{{\co\alpha}} \& \& C \arrow[tail,thick]{d}{\co\gamma} \arrow[dashed,thick]{rr}{\id} \&\& C \arrow[tail,thick]{d}[swap]{\co\gamma}\\
\lft FA \arrow[bend right = 20]{rrrr}[swap]{\radj FHf = \radj F \co\gamma' \circ \radj F\ladj F f}
\arrow{rr}{\lft Ff} \&\& \lft FC \arrow{r}{\gamma} \& C \arrow{r}{\co\gamma} \& \lft FC
\end{tikzar}
\]
On the left and on the right are the projectors $\radj F H\alpha$ and $\radj F H\gamma$. By \eqref{eq:spliteq}, splitting them gives $KH\alpha \cong \alpha$ and $KH\gamma \cong \gamma$. On the top and on the bottom is the projector morphism $\radj F Hf$. Also by \eqref{eq:spliteq}, it induces the morphism $KHf \in \Karme\AAa F(\alpha, \gamma)$ between the splittings of $\radj F H\alpha$ and $\radj F H\gamma$. It is denoted by the dashed arrow through the middle. We show that $KHf = f$. Since the projector splittings are given as the epi-mono factorizations, $KHf$ is the unique morphism from $A$ on the left to $C$ on the right making the rectangle above it and the rectangle below it commute. But the vertical arrow $\lft FC\epi C \mono \lft FC$ is also an epi-mono factorization, and 
\begin{itemize}
\item $f$ is the unique morphism making the upper left rectangle and the lower left rectangle commute (the latter by Lemma~\ref{lemma:projmorphisms});
\item $\id$ is the unique morphism making the upper right rectangle and the lower right rectangle commute (because $\gamma\circ \co\gamma = \id$). 
\end{itemize}
Hence $KH f = f$.
\end{proof}

\section{Equivalences between algebras and coalgebras}\label{Sec:ProjCoalg}

In this Section we prove the main theorem, establishing the equivalence between projective algebras and all coalgebras, and state the dual theorem, establishing the equivalence between injective coalgebras and all algebras. The equivalences, however, require the \emph{consistent}\/ homomorphisms, as in  \eqref{eq:consistent}, and not the compliant homomorphisms, like in \eqref{eq:compliant} and $\Karmo \AAa F$.

\subsection{Consistent homomorphisms}
We define the category of projective $\lft F$-algebras and consistent homomorphisms in two forms, one over the projectors in $\Klm \AAa F$, one as a subcategory of $\Emm \AAa F$ again. The first version is:
\bea\notag
|\Karm \AAa F| & = & \coprod_{X\in |\AAa|} \Big\{\varphi \in \BBb(\ladj F X, \ladj F X) \ \big|\ \varphi\circ \varphi = \varphi \wedge \lft FX \stackrel{\radj F \varphi}\looparrowright X \Big\}
\\
\Karm \AAa F (\varphi, \psi) & = & \left\{f \in \AAa(X,Y)\ \Big|\ \ 
\begin{tikzar}[row sep=1.8em,column sep=1.8em]
\ladj F X \ar{r}{\ladj F f} \ar{d}{\varphi} \& \ladj F Y \ar{d}[swap]{\psi}  \\ 
\ladj F X \ar{r}{\ladj F f}\& \ladj F Y 
\end{tikzar}\ 
\right\}\label{eq:consistent-form}
\eea
where $\lft FX \stackrel{\radj F \varphi}\looparrowright X$ is the notation from Def.~\ref{def:splitting}, meaning that $\radj F \varphi$ splits in the form $\lft FX \epi X \mono \lft FX$. The second version is:
\bea\notag
\lvert\Karmex \AAa F\rvert & = & 
\coprod_{\alpha\in | \Emm \AAa F|} \Big\{{\co\alpha}\in \Emm \AAa F (\alpha, \mu_A)\ \big|\ \alpha \circ{\co\alpha} = \id_A\Big\}\\ 
\Karmex \AAa F\left(\alpha, \gamma\right) & = & \left\{ f \in \Emm \AAa F(\alpha,\gamma )\ \Big|\ \ \begin{tikzar}[row sep=1.8em,column sep=1.8em]
 \lft F A  \ar{r}{\lft F f}
 \& \lft FC  \\
A \ar{u}[swap]{\co\alpha} \ar{r}[swap]{f}  \& C \ar{u}{\co\gamma}
\end{tikzar}\  
\right\} \label{eq:EM-projj}
\eea

\noindent\paragraph{Remarks.} Note that the $\Karmex \AAa F$-morphisms in \eqref{eq:EM-projj} are required to be in $\Emm \AAa F$, and thus satisfy the requirement of \eqref{eq:EM-proj}; but that they are \emph{moreover}\/ required to commute with the splitings $\co\alpha$ and $\co\gamma$. Concerning the $\Karm \AAa F$-homomorphisms, note that the intuitive distinction between compliant  and consistent databases from (\ref{eq:compliant}--\ref{eq:consistent}), has now been promoted to the formal distinction between the compliant homomorphisms defined in \eqref{eq:compliant-form} and the consistent homomorphisms defined in  \eqref{eq:consistent-form}. A compliant $h$ lives in $\BBb$ and satisfies $h = \psi\circ h \circ \varphi$, whereas a consistent $f$ lives in $\AAa$ and satisfies $\psi \circ \ladj F f = \ladj Ff \circ \varphi$. As for the objects, note that  in $\Karmo \AAa F$ we did not require that every projector $\ladj F X \tto \varphi \ladj F X$ satisfies $\lft FX \stackrel{\radj F \varphi}\looparrowright X$, which is required in the definition of  $\Karm \AAa F$ above. The reason is that in $\Karmo \AAa F$, every projector $\ladj F X \tto \varphi \ladj F X$ is isomorphic to a projector $\ladj F A\tto{\co \alpha '} \ladj FA$, induced by a projective $\lft F$-algebra $\lft FA\tto \alpha A$. This latter projector always satisfies the requirement $\lft FA \stackrel{\radj F \co\alpha'}\looparrowright A$, because $\radj F\alpha '$ splits into $\lft FA \eepi\alpha A \mmono{\co\alpha} \lft F A$, as proved in Lemma~\ref{lemma:incl-proj}. This isomorphism $\varphi \cong \co\alpha '$ was spelled out in the proof of Prop.~\ref{prop:eqproj}, leading to the natural isomorphism $\varphi \cong HK\varphi$. However, this isomorphism is generally not consistent: it is present in $\Karmo \AAa F$, but not in $\Karm \AAa F$. This is why  the requirement $\lft FA \stackrel{\radj F \varphi}\looparrowright A$ needs to be explicitly imposed on the objects of $\Karm \AAa F$, if the equivalence $\Karme\AAa F \simeq \Karmo \AAa F$ from Prop.~\ref{prop:eqproj} is to be extended to $\Karmex\AAa F \simeq \Karm \AAa F$.

\begin{proposition}\label{prop:eqproj2}
There is an equivalence of categories
\[
\begin{tikzar}{}
\Karmex \AAa F \arrow[bend left = 7]
{rrr}{H}[swap]{\simeq} \&\&\& \Karm \AAa F   \arrow[bend left = 7]
{lll}{K} 
\end{tikzar}
\]
where the object parts of both functors are as defined in Prop.~\ref{prop:eqproj}, whereas the arrow parts are
\[
\prooftree
f \in \Karmex \AAa F(\alpha,\gamma)
\justifies
Hf = f \in \Karm \AAa F(\co\alpha', \co \gamma')
\endprooftree
\quad \mbox{and} \quad  
\prooftree
{h}  \in \Karm \AAa F (\varphi, \psi)
\justifies
Kh\ \in\ \Karmex \AAa F\big(K\varphi, K\psi\big)
\endprooftree
\]
\end{proposition}

\begin{proof}
We begin like in the proof of Prop.~\ref{prop:eqproj}: towards the isomorphism $HK\varphi \cong \varphi$, we first split the projector $\mu_X\tto{\radj F\varphi}\mu_X$ in $\Karmex \AAa F$ to $\mu_X\eepi \alpha \alpha \mmono{\co\alpha} \mu_X$. This time, however, the assumption $\lft FA \stackrel{\radj F \co\alpha'}\looparrowright A$ means that $\radj F \varphi$ has a splitting in the form $\lft FX \eepi q X \mmono i \lft FX$. The carrier of the algebra $\alpha$ must be isomorphic to $X$, and can be chosen to be $X$ itself. Since $\mu_X$ is a free algebra, it has a unique $\lft F$-algebra homomorphism to $\alpha$. Since both $\alpha$ and $q$ are homomorphisms $\mu_X\to \alpha$, it follows that $q=\alpha$. Since each component of a splitting determines the other one, it follows that $i = \co\alpha$. Hence $\radj F\varphi = \co\alpha \circ \alpha$. It follows that $\varphi = \co \alpha '$, because
\[\varphi ' = \radj F\varphi \circ \eta = \co \alpha \circ \alpha \circ \eta = \co \alpha\]
Thus $HK\varphi = \co \alpha' = \varphi$. The natural isomorphism $KH\alpha \cong \alpha$ is constructed like in Prop.~\ref{prop:eqproj}. The only additional observation is that the condition defining the consistent morphisms in $\Karm \AAa F$ and the condition defining the inclusion preserving $\lft F$-algebra homomorphisms in $\Karmex \AAa F$ are each other's adjunction tranpose.
\end{proof}

\subsection{Projective algebras as coalgebras}\label{Sec:main}
\begin{theorem}\label{thm:main}
For every adjunction $\adj F :\BBb \to \AAa$, with the induced monad $\lft F = \radj F \ladj F$ and comonad $\rgt F = \ladj F \radj F$, the category of $\rgt F$-coalgebras is equivalent with the category of projective $\lft F$-algebras and consistent homomorphisms, provided that $\BBb$ is Cauchy complete. The equivalence is given by the functors
\[
\begin{tikzar}{}
\Emc \BBb F \arrow[bend left = 7]
{rrr}{R}[swap]{\simeq} \&\&\& \Karm \AAa F   \arrow[bend left = 7]
{lll}{L} 
\end{tikzar}
\]
where the rules
\[
\prooftree
\left(B\tto{\beta} \rgt F B\right) \in \big|\Emc \BBb F\big|
\justifies
R\beta = \left<\radj FB,\ \rgt F B \tto \varepsilon B \tto \beta \rgt F B\right> \in \big|\Karm \AAa F\big|
\endprooftree
\]
and  
\[
\prooftree
\left(B\tto{g} D\right) \in \Emc \BBb F (\beta, \delta)
\justifies
Rg =  \left(\rgt F B \tto{\rgt F g} \lft FD\right) \ \in\ \Karm \AAa F\big(R\beta, R\delta\big)
\endprooftree
\]
define $R$, whereas the object part of $L$
\[
\prooftree
\left<A,\ \ladj F A\tto \varphi \ladj F A\right> \in \big| \Karm \AAa F \big|
\justifies
L\varphi = \left(B\tto{i} \ladj F A \tto{\ladj F q'} \ladj F \radj F B\right) \in \big|\Emc \BBb F\big|
\endprooftree 
\]
and its arrow part
\[
\prooftree
\left(A\tto{f} C\right) \in \Karm \AAa F (\varphi, \psi)
\justifies
\left(B\tto{Lf} D\right)\ \in\ \Emc \BBb F\big(L\varphi, L\psi\big)
\endprooftree
\]
are defined using the projector splittings in the following diagram 
\beq\label{eq:arrowL}
\begin{tikzar}[row sep=3em,column sep=2.6em]
\ladj F \radj FB \arrow{d}{\ladj F\radj F Lf} \arrow{r}[description]{\sim}{\ladj F q''}\& \ladj FA  \arrow[two heads]{r}[swap]{q} \arrow[bend left = 20]{rr}{\varphi}  \ar{d}{\ladj Ff} \& B \arrow[bend left = 25,crossing over]{rr}{L\varphi}  \arrow[dashed]{d}[swap]{Lf} \arrow[tail]{r}[swap]{i} \& \ladj FA \arrow{d}[swap]{\ladj Ff} \arrow{r}[description]{\sim}{\ladj Fq'} \& \ladj F \radj FB \arrow{d}[swap]{\ladj F\radj F Lf } \\
\ladj F \radj FD  \arrow{r}[description]{\sim}[swap]{\ladj F p"}\&\ladj FC \arrow[two heads]{r}{p} \arrow[bend right = 20]{rr}[swap]{\psi} \& D \arrow[bend right = 25,crossing over]{rr}[swap]{L\psi} \arrow[tail]{r}{j} \& \ladj F C \arrow{r}[description]{\sim}[swap]{\ladj Fp'} \& \ladj F \radj FD \end{tikzar}
\eeq
where $q'$ is the transpose of $q$ in the projector splitting $\varphi = \left(\ladj FA \eepi q B \mmono i \ladj FA\right)$, and $q''$ is the inverse of $q'$, as explained above \eqref{eq:qprime}; and where $p'$ and $p''$ are derived from $\psi$ in a similar way.
\end{theorem}

\begin{proof}
The functor $R$ is well defined, i.e. it lands in $\Karm \AAa F$, because $\varepsilon\circ \beta = \id_B$, which implies that $R\beta$ is a projector:
\[R\beta \circ R\beta = \beta \circ \varepsilon\circ \beta \circ \varepsilon = \beta \circ \varepsilon = R\beta\] 
The fact that $Rg = \left(\rgt FB\tto{\rgt Fg} \rgt FD\right)$ is a consistent $\Karm \AAa F$-morphism follows from the naturality of $\varepsilon$ and the fact that $B\tto g D$ is an $\rgt F$-coalgebra homomorphism.

To show that the functor $L$ is well defined, we need to prove that $L\varphi = \ladj F q'\circ i$ is an $\rgt F$-coalgebra, and that $Lf$, as defined in \eqref{eq:arrowL}, is an $\rgt F$-coalgebra homomorphism. The former requirement means that $L\varphi$ must satisfy the coalgebra equations:
\bea\label{eq:Lvarphi-epsilon}
\varepsilon \circ L\varphi & = &  \id_B\\
\ladj F \radj F L\varphi \circ L\varphi & = & \ladj F \eta \circ L\varphi \label{eq:Lvarphi-nu}
\eea
To spell this out, consider diagram \eqref{eq:arrowL}. The object $B$ is defined by the splitting  $\ladj FA \eepi q B \mmono i \ladj FA$ of the projector $\varphi$; the object $D$ is defined by the splitting $\ladj FC \eepi p D \mmono j \ladj FC$ of the projector $\psi$. On the other hand, using the equivalence of categories $\Karmex \AAa F \tto{\ H\ }\Karm \AAa F$ from Prop.~\ref{prop:eqproj2}, we can assume without loss of generality that $\varphi  = H\alpha = \co\alpha'$ and $\psi = H\gamma = \co\gamma'$ for some projective algebras $\alpha, \gamma \in |\Karmex \AAa F|$, where the inclusions $\alpha \mmono{\co\alpha} \mu_A$ and $\gamma \mmono{\co\gamma}\mu_C$ transpose to $\ladj F A \tto{\co\alpha'}\ladj F A$ and $\ladj F C \tto{\co\gamma'}\ladj F C$. Lemma~\ref{lemma:incl-proj} says that the projector $\radj F \varphi = \radj F \co\alpha'$ splits into $\radj F \ladj F A \eepi{\co\alpha}A \mmono \alpha\radj F \ladj F A$. Since every functor preserves projector splittings, the $\radj F$-image of the splitting $\ladj FA \eepi q B \mmono i \ladj FA$ of $\varphi$ is also a splitting of $\radj F \varphi$. The two splittings of $\radj F \varphi$ induce an isomorphism $A\cong \radj F B$.  By chasing the following diagram, we see that this isomorphism is the transpose $A\tto{q'} \radj F B$ of $\ladj FA \tto q B$.
\beq\label{eq:qprime}
\begin{tikzar}{}
\& \radj F B \arrow[tail,bend left = 20]{dr}{\radj Fi} \\
\radj F \ladj F A \arrow[two heads,bend left = 20]{ru}{\radj F q} \arrow[two heads,bend right = 20]{rd}[swap]{\alpha} \& \& \radj F\ladj FA \\
\&A \arrow[dashed]{uu}{q'} \arrow[tail,bend right = 20]{ur}[swap]{\co\alpha} \arrow[tail,bend right = 10]{ul}[swap]{\eta}\end{tikzar}
\eeq
A chase of a similar diagram, extended by $\radj F \ladj F A \tto{\co\alpha} A$ on the right, shows that the inverse of $A\tto{q'} \radj FB$ is the composite $q'' = \left(\radj FB \tto{\radj F i} \radj F \ladj F A \tto\alpha A\right)$.
But already the fact that $q'$ is an isomorphism with an inverse $q''$ yields the transposition  
\[
\prooftree 
\left(\radj F B \tto{q''} A \tto{q'} \radj F B\right) = \id_{\radj FB}
\justifies
\left(\ladj F\radj F B \tto{\ladj F q''} \ladj FA \tto q B\right) = \varepsilon_B
\endprooftree
\]
As an extension of the projector splitting $\varphi = i\circ q$ along the isomorphism $\ladj Fq'$ on the right and along its inverse $\ladj Fq''$ on the left, the composite $L\varphi \circ \varepsilon$ is clearly a projector splitting. Hence \eqref{eq:Lvarphi-epsilon}.

Towards \eqref{eq:Lvarphi-nu}, consider the following split equalizer\footnote{The commutativity convention is explained in Appendix B.} in $\Karmo \AAa F$
\beq\label{eq:spliteq-three}
\begin{tikzar}{}
\varphi
\arrow[tail]{rr}{\ladj F \eta \circ \varphi}  
\&\&  
\ladj F\radj F\varphi
\arrow[shift left=.75ex]{rr}{\ladj F\radj F(\ladj F \eta \circ\,  \varphi)}
\arrow[tail,shift right=.75ex]{rr}[swap]{\ladj F \eta\, \circ\,  \ladj F \radj F \varphi}
\arrow[bend left = 25,two heads]{ll}{\varphi \circ\varepsilon}
\&\&  
\ladj F\radj F \ladj F\radj F\varphi 
\arrow[shift left = 1ex,bend left = 35,two heads]{ll}{\ladj F \radj F \varphi\,  \circ\, \varepsilon} 
\end{tikzar}
\eeq
Splitting the projectors in $\BBb$ yields the following split equalizer
\beq
\begin{tikzar}{}
B
\arrow[tail]{rr}{L\varphi}  
\&\&  
\ladj F\radj F B
\arrow[shift left=.75ex]{rr}{\ladj F\radj FL\varphi }
\arrow[tail,shift right=.75ex]{rr}[swap]{\ladj F \eta}
\arrow[bend left = 25,two heads]{ll}{\varepsilon} 
\&\&  
\ladj F\radj F \ladj F\radj FB
\arrow[shift left = 1ex,bend left = 35,two heads]{ll}{\varepsilon} 
\end{tikzar}
\eeq
which gives \eqref{eq:Lvarphi-nu}.  

The same reasoning applied to $\psi$ and its splitting in \eqref{eq:arrowL} shows that $L\psi$, as defined there, is also an $\rgt F$-coalgebra. Combining \eqref{eq:qprime} with the analogous diagram for $\radj F\psi$, splitting into $\gamma$ and $\co\gamma$, furthermore gives 
\bear
\radj F Lf = p' \circ f \circ q''
\eear
The definition of $Lf$ in \eqref{eq:arrowL} then displays the equation $\rgt F f\circ L\varphi = L\psi \circ f$, which means that $Lf$ is an $\rgt F$-coalgebra homomorphism from $L\varphi$ to $L\psi$. This completes the proof that the functors $R$ and $L$ are well defined.

To see that the counit $e : LR\to \Id$ is a natural isomorphism, set in \eqref{eq:arrowL} $\varphi = R\beta$ and $A=\radj FB$, which makes $q' = q'' = \id$. Since the definition of $R$ gives the splitting $R\beta = \beta \circ \varepsilon$, and the definition of $L$ says that $LR\beta$ is  the monic part of that splitting (followed by $\ladj F q'$, which is now identity), we have $LR \beta = \beta$. The counit $e : LR\to \Id$ is thus the identity.

The fact that $h:\Id \to RL$ is a natural isomorphism can be seen on \eqref{eq:arrowL}, which displays not only $\varphi$ and $L\varphi = \ladj Fq' \circ i$, but also $RL\varphi = L\varphi \circ \varepsilon = ( \ladj Fq' \circ i) \circ (q \circ \ladj F q'')  = \ladj Fq' \circ \varphi \circ \ladj F q''$. The isomorphism $h_\varphi \in \Karm \AAa F(\varphi, RL\varphi)$ is thus $h_\varphi  = \left( A\tto{q'}\radj F B\right)$ in $\AAa$. The fact that it is consistent, i.e. an $\Karm \AAa F$-morphism, is clear from the following commutative diagram.
\beq
\begin{tikzar}[row sep=1.5em,column sep=1.5em]
\ladj FA \arrow{rrr}{h_\varphi = \ladj F q'}[description]{\sim} \arrow[bend right=25]{ddd}[swap]{\varphi} \&\&\&\ladj F\radj F B \arrow{d}[swap]{\ladj F q''} \arrow[bend left = 25]{ddd}{RL\varphi}\\ 
\&\&\& \ladj FA \arrow{d}[swap]{\varphi}  \\
\&\&\& \ladj FA \arrow{d}[swap]{\ladj Fq'} \\
\ladj FA \arrow{rrr}[swap]{\ladj Fq'}[description]{\sim} \&\&\& \ladj F\radj F B  
\end{tikzar}
\eeq
This completes the proof that $L\dashv R: \Emc \BBb F \to \Karm \AAa F$ is an equivalence of categories.
\end{proof}

\subsection{Injective coalgebras as algebras}
As it is usually the case with algebras and coalgebras, the dual constructions are symmetric, but their interpretations and concrete applications are quite different. For the moment, we just spell out the dual structures and propositions, and leave the dual proofs as an exercise.

While every algebra is a quotient of a free algebra, and projective algebras are also subalgebras of free algebras, every coalgebra is a subalgebra of a cofree coalgebra, and \emph{injective}\/ coalgebras are also quotients of cofree coalgebras. The category of injective $\rgt F$-coalgebras and compliant homomorphisms is thus
\bear
|\Karco \BBb F| & = & \coprod_{X\in |\BBb |} \Big\{\varphi \in \AAa(\radj F X, \radj F X)\ \Big| \ \varphi\circ \varphi = \varphi\wedge \ladj F \radj FX \stackrel{\radj F \varphi}\looparrowright X\Big\}\\
\Karco \BBb F (\varphi, \psi) & = & \left\{h \in \AAa(\radj FX, \radj F Y)\ \Big|\ \ 
\begin{tikzar}[row sep=1.8em,column sep=1.8em]
\radj F X \ar{r}{h} \ar{d}{\varphi} \& \radj F Y  \\ 
\radj F X \ar{r}{h}\& \radj F Y \ar{u}{\psi}
\end{tikzar}\ 
\right\}
\eear
On the other hand, the category of injective coalgebras and consistent homomorphisms is
\bear
|\Karc \BBb F| & = & \coprod_{X\in |\BBb|} \Big\{\varphi \in \AAa(\radj F X, \radj F X) \ \big|\ \varphi\circ \varphi = \varphi \Big\}
\\
\Karc \BBb F (\varphi, \psi) & = & \left\{f \in \BBb(X, Y)\ \Big|\ \ 
\begin{tikzar}[row sep=1.8em,column sep=1.8em]
\radj F X \ar{r}{\radj F f} \ar{d}{\varphi} \& \radj F Y \ar{d}[swap]{\psi}  \\ 
\radj F X \ar{r}{\radj F f}\& \radj F Y 
\end{tikzar}\ 
\right\}
\eear

\begin{theorem}\label{thm:comain}
For every adjunction $\adj F :\BBb \to \AAa$, with the induced monad $\lft F = \radj F \ladj F$ and comonad $\rgt F = \ladj F \radj F$, the category of $\lft F$-algebras is equivalent with the category of injective $\rgt F$-algebras and consistent homomorphisms. The equivalence is given by the functors
\[
\begin{tikzar}{}
\Emm \AAa F \arrow[bend left = 7]
{rrr}{R}[swap]{\simeq} \&\&\& \Karc \BBb F   \arrow[bend left = 7]
{lll}{L} 
\end{tikzar}
\]
where the rules
\[
\prooftree
\left(\lft A\tto{\alpha} A\right) \in \big|\Emm \AAa F\big|
\justifies
R\alpha = \left<\ladj FA,\ \lft F A \tto \alpha A \tto \eta \lft F A\right> \in \big|\Karc \BBb F\big|
\endprooftree
\]
and  
\[
\prooftree
\left(A\tto{f} C\right) \in \Emm \AAa F (\alpha, \gamma)
\justifies
Rf =  \left(\lft F A \tto{\lft F f} \lft FC\right) \ \in\ \Karc \BBb F\big(R\alpha, R\gamma\big)
\endprooftree
\]
define $R$, whereas the object part of $L$
\[
\prooftree
\left<B,\ \radj F B\tto \varphi \radj F B\right> \in \big| \Karc \BBb F \big|
\justifies
L\varphi = \left(\radj F \ladj F A\tto{\radj F i'} \radj F B \tto{q} A\right) \in \big|\Emm \AAa F\big|
\endprooftree 
\]
and its arrow part
\[
\prooftree
\left(B\tto{g} D\right) \in \Karc \BBb F (\varphi, \psi)
\justifies
\left(A\tto{Lg} C\right)\ \in\ \Emm \AAa F\big(L\varphi, L\psi\big)
\endprooftree
\]
are defined using the projector splittings in the following diagram 
\beq\label{eq:arrowLL}
\begin{tikzar}[row sep=3em,column sep=2.6em]
\radj F \ladj FA \arrow[bend left = 35]{rr}{L\varphi} \arrow{d}{\radj F\ladj F Lg} \arrow{r}[description]{\sim}{\radj F i'}\& \radj FB  \arrow[two heads]{r}[swap]{q} \arrow[bend left = 20,crossing over]{rr}{\varphi}  \ar{d}{\radj Fg} \& A   \arrow[dashed]{d}[swap]{Lg} \arrow[tail]{r}[swap]{i} \& \radj FB \arrow{d}[swap]{\radj Fg} \arrow{r}[description]{\sim}{\radj Fi''} \& \radj F \ladj FA \arrow{d}[swap]{\radj F\ladj F Lg } \\
\radj F \ladj FC \arrow[bend right = 25]{rr}[swap]{L\psi}   \arrow{r}[description]{\sim}[swap]{\radj F j'}\&\radj FD \arrow[two heads]{r}{p} \arrow[bend right = 20,crossing over]{rr}[swap]{\psi} \& C \arrow[tail]{r}{j} \& \radj F D \arrow{r}[description]{\sim}[swap]{\radj Fj''} \& \radj F \ladj FC \end{tikzar}
\eeq
\end{theorem}

\section{Application}\label{Sec:Application}
Interpreted along the lines of the example from Sec.~\ref{sec:example}, the category $\Karm \CCc S$ of projectors and consistent homomorphisms, induced by the state monad $\lft S$ on a cartesian closed category $\CCc$, can be viewed as a model of data release policies. The idea is that
 \begin{itemize}
\item a projector $\left(S\times A \tto \varphi S\times A\right) \in \left| \Karm \CCc S\right|$ filters private data $a\in A$ and private states $s\in S$ and releases  a public pair 
\bear
\varphi (s,a) & = & \left<\varphi_0(s,a), \varphi_1(s,a)\right> \in S\times A
\eear
\item a morphism $\left(A\tto f B\right) \in \Karm \CCc  S(\varphi, \psi)$ can be thought of as a deterministic channel which maps data of type $A$ to data of type $B$ in such a way that the following requirements are satisfied
\bea\label{eq:consone}
\psi_0\left(s, f(a)\right) & = & \varphi_0(s,a)\\
\psi_1\left(s, f(a)\right) & = & f\left(\varphi_1(s,a)\right)\label{eq:constwo}
\eea
where $\left(S\times B \tto \psi S\times B\right) \in \left|\Karm \CCc S\right|$ is another policy.
\end{itemize}
Conditions (\ref{eq:consone}--\ref{eq:constwo}) guarantee that the channel $f$ behaves consistently with the policies $\varphi$ and $\psi$. Note that this is a special case of the model from Sec.~\ref{sec:example}, in the sense that we are not capturing the consistency of a database $S\times A \tto g S\times B$ with the policies $S\times A \tto \varphi S\times A$ and $S\times B \tto \psi S\times B$, but only the consistency of a stateless channel $A\tto f B$.

If we accept this restriction for a moment, the equivalence $\Karm \CCc S \simeq \Emc \CCc S$ provides an interesting characterization of data release policies, with the consistent channels as the morphisms between them: they are equivalent to coalgebras for the comonad
\bea
\rgt S :\  \CCc & \to & \CCc\\
X & \mapsto & S\times (S\Rightarrow X) \notag
\eea
with the coalgebra homomorphisms. More precisely, a projective $\lft S$-algebra $S\times A\tto \varphi S\times A$, which we viewed as a Mealy machine \cite[Sec.~2.7(a)]{Hopcroft-Ullman:book}
\beq
S\times A \tto{\varphi_0} S \qquad \qquad S\times A \tto{\varphi_1} A
\eeq
induces a coalgebra $B \tto\beta S\times (S\Rightarrow B)$, which boils down to a pair of maps reminding of a Moore machine \cite[Sec.~2.7(b)]{Hopcroft-Ullman:book} 
\beq\label{eq:moore}
B\tto{\beta_0} S \qquad \qquad B\times S\tto{\beta_1} B
\eeq
It is conspicuous, however, that the state space $S$ of the Mealy machine $\varphi$ has become the alphabet in the corresponding Moore machine $\beta = L\varphi$, where $L$ is the functor from Thm.~\ref{thm:main}. The state space $B$ of the Moore machine $\beta$ arises from the construction of $L\varphi$ in \eqref{eq:arrowL} as the set of public pairs:
\bear
B & = & \left\{<s,a>\in S\times A\ |\ \varphi(s,a) = <s,a>\right\}
\eear
Note, however, that both machines are of a very special kind: the Mealy machine is idempotent, and the Moore machine satisfies the coalgebra conditions
\[
\varepsilon \circ \beta = \id \qquad \qquad \rgt S \beta \circ \beta = \nu \circ \beta
\]
which for the components in \eqref{eq:moore} correspond to the following equations
\bear
\beta_0\left(\beta_1(b,s)\right) & = & s\\
\beta_1\left(b, \beta_0(b)\right) & = & b\\
\beta_1\left(\beta_1(b,s),t) \right) & = & \beta_1(b,t)
\eear
In a sense (formalized by Thm.~\ref{thm:main}), these equations realize on the set of public pairs $B$ precisely the data filtering condition that was realized on the set of all pairs $S\times A$  by the idempotence of $\varphi$.

To go beyond the stateless morphisms, and capture not just channels in the form $A\tto f B$, but also databases in the form $S\times A \tto g S\times B$, consider the adjunction 
\beq\label{eq:free}
\begin{tikzar}{}
\CCc \arrow[bend left = 13]{rr}{S^\flat}[swap]{\bot} \&\& \Klm \CCc S  \arrow[bend left = 13]{ll}{S_\flat} 
\end{tikzar}
\eeq
where 
\begin{align*}
S^\flat X & = X && S_\flat X = \lft SX\\
S^\flat f & = (S\times f) && S_\flat f  = (S\Rightarrow f)
\end{align*}
The category of injective coalgebras in $\Klm \CCc S$ is now as follows:
\bear
|\Klm \CCc S ^\looparrowright| & = & \coprod_{A\in |\CCc|} \Big\{\varphi \in \CCc(\lft S A, \lft S A)\ \Big| \ \varphi\circ \varphi = \varphi\Big\}\\
\Klm \CCc S ^\looparrowright (\varphi, \psi) & = & \left\{g \in \CCc(S\times A, S\times B)\ \Big|\ \ 
\begin{tikzar}[row sep=1.8em,column sep=2.2em]
\lft S A \ar{r}{S\Rightarrow g} \ar{d}{\varphi} \& \lft S B  \ar{d}[swap]{\psi}  \\ 
\lft S A \ar{r}[swap]{S\Rightarrow g}\& \lft S B
\end{tikzar}\ 
\right\}
\eear
The consistent homomorphisms $S\times A \tto g S\times B$ can now be construed as databases. The policies with which they are consistent are more general than those considered so far. Policies $S\times A\to S\times A$ in $\Karm \CCc S$ filtered private states and data and supplied the corresponding public pairs. Policies $\lft S A\to \lft S A$ in $\Klm \CCc S ^\looparrowright$ filter entire stateful behaviors. 

Interestingly, though, Thm.~\ref{thm:comain} provides the equivalence $\Klm \CCc S ^\looparrowright \simeq \Emm \CCc S$, and one of its corollaries\footnote{to be presented elsewhere} provides an equivalence $\Emm \CCc S \simeq \CCc$. The equivalences are nontrivial, and may require further research. They suggest that implementing policies within a model of data release, and using these policies to filter out the private data, and to extract the public data alone, leads to an equivalent model, but this time consisting of the public data alone. Filtering out the private data can thus be formalized as an equivalence. Privacy policies can be formalized to make the publicly released data structurally indistinguishable from all data.

\section{Related and further work}\label{Sec:Closing}
Although the presented constructions emerged within a practice-driven effort towards modeling and analyzing data release policies using the salient tools of monadic programming, the research path led through the realm of basic monad theory, with some old questions still lurking, and with the theoretic repercussions surpassing not only our practical goals, but probably also our current understanding. 
Back in 1968, in the first of the Batelle volumes,   Barr \cite{BarrM:algCoalg} raised the question of comonadicity of the left adjoint of a monadic functor. More precisely, he considered the adjunctions in the form
\beq\label{eq:barr}
\begin{tikzar}{}
\AAa \arrow[bend left = 13]{rr}{T^\sharp}[swap]{\bot} \&\& \Emm \AAa T  \arrow[bend left = 13]{ll}{T_\sharp} 
\end{tikzar}
\eeq
and asked under which conditions would the functor $T^\sharp$ be comonadic. This means that the comparison functor
\[
\begin{tikzar}{}
\AAa \arrow{r} \& \Emc{\left(\Emm \AAa T\right)} T  \end{tikzar}
\]
should be an equivalence, where $\rgt T =  T^\sharp T_\sharp$. Barr provided the answer for the special cases when $\AAa$ is the category of sets, pointed or not, and when it is the category of vector spaces and linear operators. He suggested that the general answer might be difficult. 
The question seems to have been reemerging regularly in various guises, most recently in Jacobs' work on coalgebras over algebras as an abstract form of the concept of basis \cite{JacobsB:bases}, extending the results of \cite{PavlovicD:MSCS13} about bases as Sweedler-style coalgebras to bases as coalgebras for comonads.  

Theorem~\ref{thm:main} provides the equivalence $\Emc \BBb F \simeq \Karm \AAa F$ for \emph{any}\/ resolution $\adj F :\BBb \to \AAa$ of the monad $\lft F$. More precisely, if besides the adjunction $\adj F :\BBb \to \AAa$ there is also $\adj G : \CCc \to \AAa$, and the monads $\radj G \ladj G \cong \radj F \ladj F$ are isomorphic to a monad $T:\AAa \to \AAa$, coherently with respect to the monad structures \cite[Sec.~3.6]{BarrM:TTT}, then
\[ \Emc \BBb F\ \simeq\ \Karm \AAa T\ \simeq\ \Emc \CCc G \]
Instantiating to the adjunction displayed in \eqref{eq:barr} gives
\bear
\Emc{\left(\Emm \AAa T\right)} T &\simeq & \Karm \AAa T
\eear
The question of comonadicity of $T^\sharp$ can thus be studied on the comparison functor
\[
\begin{tikzar}{}
\AAa \arrow{r} \& \Karm \AAa T  \end{tikzar}
\]
which the reader may enjoy as an exercise. In a similar way, Beck's General Monadicity Theorem \cite[Thm.~3.3.13]{BeckJ:thesis,BarrM:TTT} can be stated and proved by unravelling the projectors from behind the split coequalizers in the original formulation. In general, the projector view of algebras and coalgebras, opened by Theorems~\ref{thm:main} and \ref{thm:comain}, seems to facilitate analyses of monadicity, and even enable analysis of relative monadicity \cite{AltenkirchT:relative-formalized,AltenkirchT:relative-monads}. On the other hand, it opens up an alley towards classifying resolutions in general, as illustrated in the following diagram.
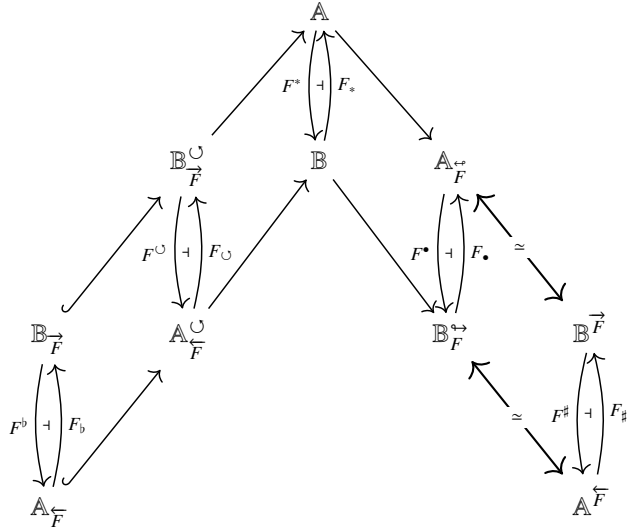
\[
\begin{tikzar}[row sep=4em,column sep=3em]
\& \& \AAa \arrow[bend right = 12]{d}{\dashv}[swap]{\ladj F} \arrow{dr}
\\
\&\Karco \BBb F 
\arrow[bend right = 12]{d}{\dashv}[swap]{F^\circlearrowleft} \arrow{ur}
\& \BBb \arrow{dr}
\arrow[bend right = 12]{u}[swap]{\radj F} \& 
\Karm \AAa F  \arrow[bend right = 12]{d}{\dashv}[swap]{F^\bullet} \arrow[leftrightarrow,thick]{dr}[description]{\simeq}  
\\
\Klc \BBb F \arrow[bend right = 12]{d}{\dashv}[swap]{\lkadj F} \arrow[hook]{ur}
\& \Karmo \AAa F 
\arrow{ur}
\arrow[bend right = 12]{u}[swap]{F_\circlearrowleft} \& \& \Karc \BBb F \arrow[bend right = 12]{u}[swap]{F_\bullet} \arrow[leftrightarrow,thick]{dr}[description]{\simeq} 
\& \Emc \BBb F  \arrow[bend right = 12]{d}{\dashv}[swap]{\lnadj{F}}
\\
\Klm \AAa F \arrow[hook]{ur}
\arrow[bend right = 12]{u}[swap]{\rkadj F} \&\& \&\& \Emm \AAa F \arrow[bend right = 12]{u}[swap]{\rnadj{F}}
\end{tikzar}
\]
Of special interest here is the adjunction $\left(F^\bullet \dashv F_\bullet\right): \Karc \BBb F \to \Karm \AAa F$ defined by
\[
\prooftree
\left<A,\ \ladj F A\tto \varphi \ladj F A\right> \in \big| \Karm \AAa F \big|
\justifies
\left<\ladj F \radj F B_\varphi,\ \  \radj F \ladj F \radj F B_\varphi \tto{\radj F \varepsilon} \radj F B_\varphi \tto{\ \eta\ } \radj F \ladj F \radj F B_\varphi\right> \in \big| \Karc \BBb F \big|
\endprooftree 
\]
and
\[
\prooftree
\left<D,\ \radj F D\tto \psi \radj F D\right> \in \big| \Karc \BBb F \big|
\justifies
\left< \radj F \ladj F C_\psi,\ \  \ladj F \radj F \ladj F C_\psi \tto{\ \varepsilon\ } \ladj F C_\psi \tto{\ladj F \eta } \ladj F \radj F \ladj F C_\psi\right>\in \big| \Karm \AAa F \big|
\endprooftree 
\]
where $\ladj FA \epi B_\varphi \mono \ladj FA$ and $\radj FD \epi C_\psi \mono \radj FD$ are the splittings of $\varphi$ and of $\psi$, respectively. The adjunction $\left(F^\bullet \dashv F_\bullet\right): \Karc \BBb F \to \Karm \AAa F$ is the \emph{nucleus}\/ of the adjunction $\adj F:\BBb\to \AAa$ \cite{PavlovicD:ICFCA12,PavlovicD:Samson13,PavlovicD:CALCO15,PavlovicD:nucl,WillertonS:tight,WillertonS:legendre}. The fact that the functor $F_\bullet$ is monadic, and the functor $F^\bullet$ is comonadic will be proved in the full version of this paper.

\bibliographystyle{plain}
\bibliography{mond-ref,CT,PavlovicD}

\appendix

\section{Appendix: Adjunctions, monads, comonads}
\begin{definition}
An adjunction $F = \left(\adj F : \BBb\to \AAa\right)$ consists of the following data
\begin{itemize}
\item categories $\AAa$ and $\BBb$
\item functors $\ladj F: \AAa \to \BBb$ and $\radj F: \BBb\to \AAa$, called \emph{left adjoint} and \emph{right adjoint}, respectively
\item natural transfromations $\eta: \id_\AAa \to \radj F \ladj F$ and $\varepsilon: \ladj F \radj F \to \id_\BBb$, called the \emph{adjunction unit} and \emph{counit}, respectively,
\end{itemize}
such that
\begin{alignat*}{5}
\begin{tikzar}
\AAa \arrow{d}{\ladj{F}} 
\arrow[bend right=75]{dd}[swap]{\id}[name=L]{} 
\\
\BBb \arrow{d}{\radj F} 
\arrow[bend left=75]{dd}{\id}[name=R,swap]{}
\arrow[Leftarrow,to path = -- (L)\tikztonodes]{}[swap]{\eta} 
\\ \AAa \arrow{d}{\ladj{F}} \arrow[Rightarrow,to path = -- (R)\tikztonodes]{}{\varepsilon}
\\
\BBb
\end{tikzar}
& =\ \  \begin{tikzar}[row sep=2.8em]
\AAa \arrow{ddd}{\ladj{F}} 
\\
\\ 
\\
\BBb
\end{tikzar}
&\qquad \qquad & 
\begin{tikzar}
\BBb \arrow{d}{\radj{F}} \arrow[bend left=75]{dd}{\id}[name=R,swap]{}\\
\AAa \arrow{d}{\ladj F} 
\arrow[bend right=75]{dd}[swap]{\id}[name=L]{} 
\arrow[Rightarrow,to path = -- (R)\tikztonodes]{}{\varepsilon}
\\ \BBb 
\arrow{d}{\radj{F}} 
\arrow[Leftarrow,to path = -- (L)\tikztonodes]{}[swap]{\eta} 
\\
\AAa
\end{tikzar}
&\ \  =\ \  \begin{tikzar}[row sep=2.8em]
\BBb \arrow{ddd}{\radj{F}} 
\\
\\ 
\\
\AAa
\end{tikzar}
\end{alignat*}
\end{definition}

\begin{definition}
A monad $(T,\eta, \mu)$ consists of the following data
\begin{itemize}
\item category $\AAa$
\item functor $T: \AAa \to \AAa$ 
\item natural transformations $\eta: \id_\AAa \to T$ and $\mu: TT\to T$, called, respectively, the  \emph{monad unit} and \emph{evaluation} (or \emph{cochain})
\end{itemize}
which satisfy the following conditions
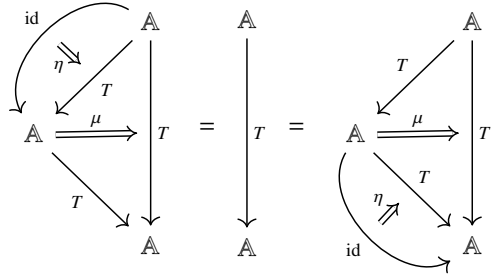
\begin{alignat*}{3}
\begin{tikzar}[row sep=2.8em,column sep=2.8em]
\AAa \arrow{d}[swap]{T} \arrow[phantom]{}[name=U,below=5]{}
\&
\AAa 
\arrow{l}[swap]T 
\arrow{dd}{T}[name=R,swap]{} 
\arrow{dl}{T}[name=D,swap]{}
\arrow[Rightarrow,to path =(U) -- (D)\tikztonodes]{}[swap]{\mu} 
\\
\AAa \arrow{dr}[swap]{T} 
\arrow[Rightarrow,to path = -- (R)\tikztonodes]{}{\mu} 
\\ \& \AAa 
\end{tikzar}
&\ \  =\ \  \begin{tikzar}[row sep=2.8em,column sep=2.8em]
\&
\AAa 
\arrow{dd}{T}[name=R,swap]{} 
\arrow{dl}[swap]{T}
\\
\AAa \arrow{dr}{T}[name=D,swap]{}
\arrow{d}[swap]{T} \arrow[Rightarrow,to path = -- (R)\tikztonodes]{}{\mu} 
\\ 
\AAa \arrow{r}[swap]T 
\arrow[phantom]{}[name=U,above=5]{}
\arrow[Rightarrow,to path =(U) -- (D)\tikztonodes]{}{\mu} 
\& \AAa 
\end{tikzar}
\end{alignat*}

\begin{alignat*}{3}
\begin{tikzar}[row sep=2.8em,column sep=2.8em]
\&
\AAa 
\arrow[bend right = 75]{dl}[swap]{\id}[name=U]{} 
\arrow{dd}{T}[name=R,swap]{} 
\arrow{dl}{T}[name=D,swap]{}
\arrow[Rightarrow,to path =(U) -- (D)\tikztonodes]{}[swap]{\eta} 
\\
\AAa \arrow{dr}[swap]{T} 
\arrow[Rightarrow,to path = -- (R)\tikztonodes]{}{\mu} 
\\ \& \AAa 
\end{tikzar}
&\ \  \ \ =\ \  \begin{tikzar}[row sep=2.4em]
\AAa \arrow{ddd}{T} 
\\
\\ 
\\
\AAa
\end{tikzar}
&\ \  \ \ =\ \ \begin{tikzar}[row sep=2.8em,column sep=2.8em]
\&
\AAa 
\arrow{dd}{T}[name=R,swap]{} 
\arrow{dl}[swap]{T}\\
\AAa 
\arrow[bend right = 75]{dr}[swap]{\id}[name=U]{} 
\arrow{dr}{T}[name=D,swap]{}
\arrow[Rightarrow,to path =(U) -- (D)\tikztonodes]{}{\eta}  
\arrow[Rightarrow,to path = -- (R)\tikztonodes]{}{\mu} 
\\ \& \AAa 
\end{tikzar}
\end{alignat*}
\end{definition}

\begin{definition}
A comonad $(S,\varepsilon, \nu)$ consists of the following data
\begin{itemize}
\item category $\AAa$
\item functor $S: \AAa \to \AAa$ 
\item natural transfromations $\varepsilon : S\to  \id_\AAa$ and $\nu: S\to SS$, called the  \emph{comonad counit} and \emph{coevaluation} (or \emph{chain}), respectively,
\end{itemize}
which satisfy the following conditions
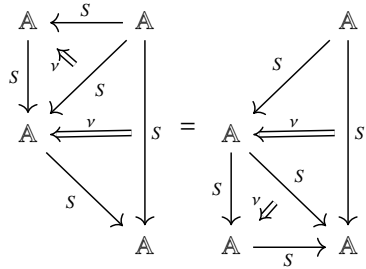
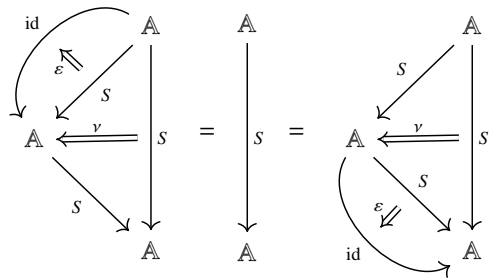
\begin{alignat*}{3}
\begin{tikzar}[row sep=2.8em,column sep=2.8em]
\AAa \arrow{d}[swap]{S} \arrow[phantom]{}[name=U,below=5]{}
\&
\AAa 
\arrow{l}[swap]S 
\arrow{dd}{S}[name=R,swap]{} 
\arrow{dl}{S}[name=D,swap]{}
\arrow[Leftarrow,to path =(U) -- (D)\tikztonodes]{}[swap]{\nu} 
\\
\AAa \arrow{dr}[swap]{S} 
\arrow[Leftarrow,to path = -- (R)\tikztonodes]{}{\nu} 
\\ \& \AAa 
\end{tikzar}
&\ \  =\ \  \begin{tikzar}[row sep=2.8em,column sep=2.8em]
\&
\AAa 
\arrow{dd}{S}[name=R,swap]{} 
\arrow{dl}[swap]{S}
\\
\AAa \arrow{dr}{S}[name=D,swap]{}
\arrow{d}[swap]{S} \arrow[Leftarrow,to path = -- (R)\tikztonodes]{}{\nu} 
\\ 
\AAa \arrow{r}[swap]S 
\arrow[phantom]{}[name=U,above=5]{}
\arrow[Leftarrow,to path =(U) -- (D)\tikztonodes]{}{\nu} 
\& \AAa 
\end{tikzar}
\end{alignat*}

\begin{alignat*}{3}
\begin{tikzar}[row sep=2.8em,column sep=2.8em]
\&
\AAa 
\arrow[bend right = 75]{dl}[swap]{\id}[name=U]{} 
\arrow{dd}{S}[name=R,swap]{} 
\arrow{dl}{S}[name=D,swap]{}
\arrow[Leftarrow,to path =(U) -- (D)\tikztonodes]{}[swap]{\varepsilon} 
\\
\AAa \arrow{dr}[swap]{S} 
\arrow[Leftarrow,to path = -- (R)\tikztonodes]{}{\nu} 
\\ \& \AAa 
\end{tikzar}
&\ \  \ \ =\ \  \begin{tikzar}[row sep=2.4em]
\AAa \arrow{ddd}{S} 
\\
\\ 
\\
\AAa
\end{tikzar}
&\ \  \ \ =\ \ \begin{tikzar}[row sep=2.8em,column sep=2.8em]
\&
\AAa 
\arrow{dd}{S}[name=R,swap]{} 
\arrow{dl}[swap]{S}\\
\AAa 
\arrow[bend right = 75]{dr}[swap]{\id}[name=U]{} 
\arrow{dr}{S}[name=D,swap]{}
\arrow[Leftarrow,to path =(U) -- (D)\tikztonodes]{}{\varepsilon}  
\arrow[Leftarrow,to path = -- (R)\tikztonodes]{}{\nu} 
\\ \& \AAa 
\end{tikzar}
\end{alignat*}

\end{definition}

\begin{definition} The Kleisli construction maps the monad $T:\AAa\to \AAa$ to the adjunction $\TKL T = \left(\kadj{T}: \Klm \AAa {T} \to \AAa\right)$ where the category $\Klm \AAa {T}$ consists of  
\begin{itemize}
\item \emph{free algebras} as objects, which boil down to $|\Klm \AAa T|\  = \ |\AAa|$; 
\item \emph{algebra homomorphisms} as arrows, which boil down to $\Klm \AAa T(x,x') = \AAa(x,Tx')$;
\end{itemize}
with the composition
\bear
\Klm \AAa T (x,x') \times \Klm \AAa T(x',x'') & \stackrel \circ \longrightarrow & \Klm \AAa T(x,x'')\\
\big < x \tto f Tx'\, ,\ x'\tto g Tx'' \big> &\longmapsto & \big(x\tto f Tx' \tto {Tg} TTx'' \tto \mu Tx'' \big)
\eear
and with the identity on $x$ induced by the monad unit $\eta : x \to Tx$
\end{definition}

\begin{definition} The Eilenberg-Moore construction maps the monad $T:\AAa\to \AAa$ to the adjunction $\TEM T = \left(\nadj{T}: \Emm \AAa {T} \to \AAa\right)$ where the category $\Emm \AAa {T}$ consists of 

\begin{itemize}
\item \emph{all algebras}\/ as objects: 
\bear |\Emm \AAa {T}| &  = & \sum_{x\in |\AAa|} \big\{\alpha \in \AAa(Tx, x)\ |\ \alpha\circ \eta = \id\ \wedge \ \alpha\circ T\alpha = \alpha \circ \mu \big\}\eear
\item \emph{algebra homomorphisms}\/ as arrows: 
\bear \Emm \AAa T (Tx\tto \alpha x, Tx'\tto \gamma x') & = & \big\{f\in \AAa(x,x')\ |\ f\circ \alpha = \gamma\circ Tf \big\}\eear
\end{itemize}
\end{definition}

\section{Appendix: Split equalizers}
Split equalizers and coequalizers  \cite{BarrM:TTT,BeckJ:thesis}  are conventionally written as \emph{partially}\/ commutative diagrams: the straight arrows commute, the epi-mono splittings compose to identities on the quotient side, and to equal idempotents on the other side.

\begin{proposition}
Consider the split equalizer diagram
\[
\begin{tikzar}{}
A
\arrow[tail]{rr}{i}  
\&\&  
B
\arrow[shift left=.75ex]{rr}{f}
\arrow[tail,shift right=.75ex]{rr}[swap]{j}
\arrow[bend left = 25,two heads]{ll}{q}
\&\&  
C
\arrow[shift left = 1ex,bend left = 35,two heads]{ll}{r} 
\end{tikzar}
\]
where
\[q\circ i = \id_A \qquad   r \circ j = \id_B \qquad f\circ r\circ f = j\circ r\circ f \]
Then 
\begin{itemize}
\item $r\circ f$ is idempotent and 
\item $i$ is the equalizer of $f$ and $j$ if and only if $i\circ q = r\circ f$.
\end{itemize}
\end{proposition}
\end{document}